\documentclass[10pt,a4paper]{article}

\usepackage{amsmath,amssymb,amsthm}
\usepackage{enumitem}
\usepackage{booktabs}
\usepackage{array}
\usepackage[hidelinks]{hyperref}
\usepackage[T1]{fontenc}
\newtheorem{theorem}{Theorem}[section]
\newtheorem{lemma}[theorem]{Lemma}
\newtheorem{proposition}[theorem]{Proposition}
\newtheorem{corollary}[theorem]{Corollary}
\theoremstyle{definition}
\newtheorem{definition}[theorem]{Definition}
\theoremstyle{remark}
\newtheorem{remark}[theorem]{Remark}

\newcommand{\cmp}{\operatorname{cmp}}
\newcommand{\obs}{\operatorname{obs}}
\newcommand{\qcmp}{\operatorname{qcmp}}
\newcommand{\D}{\mathcal D}
\newcommand{\SCL}{\operatorname{SCL}}
\newcommand{\Obj}{\mathsf O}
\newcommand{\Om}{\Omega}
\newcommand{\blank}{\square}
\newcommand{\relto}{\mathrel{\rightsquigarrow}}

\title{Finite-Monoid Compression in Syntactic Concept Lattices: Arity Hierarchies and a Pseudovariety Trichotomy}
\author{Takayuki Kuriyama\\
Independent Researcher, Tokyo, Japan\\
\texttt{growup.kuriyama@gmail.com}}
\date{}

\begin{document}
\maketitle

\begin{abstract}
Clark's syntactic concept lattice (SCL) records two-sided distributional
structure, and Wurm extended it to tuples of arbitrary finite arity.
We study \(\cmp_f(L)\), the minimum image size of a finite-monoid observation
that preserves guarded tuple substitution through arity \(f\) on the
principal layer.

For regular languages, we characterize \(\cmp_f(L)\) exactly as the least
cardinality of the codomain of an \(f\)-separating relational morphism from
the pointed syntactic monoid.  Let \(\operatorname{ch}(\mathbf V)\) denote
the least arity at which these compression numbers stabilize uniformly over
a pseudovariety \(\mathbf V\).  Our main result is the following trichotomy
of possible uniform heights:
\(\operatorname{ch}(\mathbf V)\in\{1,2,\infty\}, \qquad \operatorname{ch}(\mathbf V)=\infty \Longleftrightarrow \operatorname{Synt}(\{ab\})\in\mathbf V.\)
Thus no finite uniform compression height \(3,4,\ldots\) occurs.

The infinite case is sharp: inside
\(\langle\operatorname{Synt}(\{ab\})\rangle\), every boundary
\(d\to d+1\) admits unbounded compression gaps, and arbitrary finite strict
prefixes of the arity hierarchy are realizable.  On the finite side,
commutative monoids and bands stabilize at arity one, while every completely
regular syntactic monoid stabilizes by arity two; finite group kernels show
that the binary bound is sharp.

At unary arity, every nonempty finite simple graph is realized by an explicit
length-three language, yielding an exact chromatic-number formula and
NP-completeness of deciding \(\cmp_1(L)\le3\) for explicitly listed
length-three languages.  The structural boundary between compression heights
one and two remains open.
\end{abstract}



\section{Introduction}

\subsection{Finite compositional compression}

For a regular language, exact two-sided contextual behavior is captured by its
syntactic monoid $T_L$, the canonical finite monoid recognizing $L$.  The size
of this monoid is a standard algebraic descriptional measure: syntactic-monoid
size and its relation to automaton size have been studied systematically in
regular-language complexity~\cite{HolzerKonig2004}.  This paper starts from a
more selective question.  If exact recognition is replaced by one fixed
compositional task---guarded substitution of words and tuples into accepting
contexts---how much of $T_L$ must actually be retained?  Thus the invariant
studied below is best viewed as a task-relative relaxation of syntactic
complexity rather than as a competing recognition invariant.  Throughout the
paper, ``compression'' refers to reduction of this finite observation
codomain; it does not refer to compression of words or semigroup expressions
by straight-line programs.

There is a classical precedent for expecting such a relaxation not to be a
quotient problem.  In the minimization of incompletely specified finite-state
machines, compatibility of states need not be transitive, and the standard
reduction therefore proceeds through compatible covers rather than ordinary
equivalence-class minimization; Paull and Unger developed the classical
maximal-compatible/minimal-closed-cover framework~\cite{PaullUnger1959}, and
state reduction for incompletely specified machines is NP-complete%
~\cite{Pfleeger1973}.  The analogy here is structural, not an identification
of the two problems: our finite object is a monoid rather than a transition
cover, and multiplicative compositionality imposes constraints absent from the
classical machine problem.  But the same obstruction pattern reappears:
partial specifications naturally produce reflexive symmetric compatibility
without transitivity, so quotienting may be too rigid.  Minimum separating
DFAs provide a complementary algorithmic precedent in compositional
verification, where the size of a separator controls the size of a reusable
finite assumption~\cite{ChenFarzanClarkeTsayWang2009,
VazquezPargaGarciaLopez2016}.

The Clark--Wurm syntactic-concept construction supplies the missing ingredient:
it generates the partial specification \emph{internally from the language}.
Clark's syntactic concept lattice $\SCL_1(L)$ lifts syntactic algebra from
individual words to the Galois incidence between words and their accepting
two-sided contexts~\cite{ClarkSCL}; Wurm extends the construction to tuples of
every finite arity~\cite{WurmSCLTuples}.  Words and tuples enter these lattices
through their principal concepts.  Two distinct principal concepts create a
genuine substitution conflict precisely when their accepting intents overlap:
they can occur in one common accepting context, yet some other context
distinguishes them.  Definition~\ref{scl:def:safety} records this geometric predicate once, and the
Safety Interface Theorem~\ref{alg:thm:safety-interface} proves that its
representative-level principal-SCL, guarded word/context, congruence, and
canonical relational-morphism forms are all equivalent.  The subsequent
selector lemma upgrades the canonical relational interface to optimization over
arbitrary relational morphisms.

This gives a conceptual pipeline from exact syntactic information to
Clark--Wurm principal conflicts, and from there to a minimum finite
compositional separator.
The middle step is what makes the specification intrinsic rather than
externally supplied.  The last step is task-relative: we do not try to preserve
the whole concept lattice, only the distinctions required for substitution of
actual words and tuples.  In this sense the term ``canonical'' refers to the
conflict specification generated by the Clark--Wurm geometry, not to a claim
that there is a unique compression notion for all closed concepts.

For regular languages the resulting constraints admit an exact finite
algebraic reduction.  At arity one, equality of singleton-generated concepts
is syntactic congruence~\cite[Proposition~3.19]{ClarkSCL}, so the unary
principal layer is canonically $T_L$.  Higher SCL levels generate additional
separation constraints, while the unary principal monoid supplies a reusable
compositional palette.  We minimize the number of monoid values required
simultaneously through arity $f$; the optimum is the \emph{finite-monoid
compression number} $\cmp_f(L)$.  An optimal observation need not be constant
on syntactic classes, so the correct finite algebraic semantics is generally a
relational morphism of the pointed syntactic monoid rather than an ordinary
quotient.  This is the multiplicative analogue of the cover-like phenomenon
above, not a reduction of ISFSM minimization to our setting.

Exact recognition is always a witness, hence
\(\cmp_1(L)\le \cmp_2(L)\le\cdots\le |T_L|.\)
The chain quantifies how much exact syntactic information survives as the
required substitution arity increases.  It is already nontrivial at arity one:
for $L=\{a,a^3\}$, the contextual-conflict graph is 2-colorable but no
two-element monoid coloring is compatible with concatenation, and
$\cmp_1(L)=3$.  More sharply, Section~\ref{sec:conflict} gives the unary
family $K_n=\{a^n,a^{n+2}\}$, $n\ge2$, for which
\[
\begin{aligned}
 \chi(\Gamma_{K_n})=2<3=\cmp_f(K_n)&<n+2=\qcmp_f(K_n)\\
 &<|T_{K_n}|=n+4 \qquad(f\ge1).
\end{aligned}
\]
Thus passing from arbitrary finite compression to quotients of the syntactic
monoid can cost an unbounded number of states; the multivalued relational
interface is not merely a technical reformulation.  At the opposite structural
extreme, for a group-kernel language $K_\eta=\eta^{-1}(1_G)$ we later obtain
$\cmp_1=1$ and $\cmp_f=[G:Z(G)]$ for $f\ge2$.  Finally, every nonempty finite simple graph is
realized exactly at the unary level by an explicit length-three language.  These
examples show that $\cmp_f$ is neither syntactic-monoid size, SCL size, nor
ordinary graph coloring: it measures the least \emph{compositional} finite
resource preserving the SCL-generated substitution specification.

\subsection{Arity hierarchy and main theorem}

The organizing question is how the least compositional resource changes when
higher tuple arities are required.  Write
\(M_{ab}:=\operatorname{Synt}(\{ab\})\) and
\(\mathbf V_{ab}:=\langle M_{ab}\rangle\), where
\(M_{ab}=\{1,a,b,ab,0\}\).
At the level of individual languages, the hierarchy can remain arbitrarily
informative: inside this single five-element-generated aperiodic
pseudovariety, every boundary \(d\to d+1\) admits an unbounded multiplicative
increase in compression.  More strongly, for every finite \(m\) one language
can realize
\(\cmp_1(L)<\cmp_2(L)<\cdots<\cmp_m(L),\)
and the successive ratios can simultaneously exceed arbitrary prescribed
finite thresholds.  A triangular gluing principle isolates the mechanism
behind this finite-prefix universality.

Remarkably, uniformizing over a pseudovariety produces the opposite
phenomenon: the arbitrarily long local hierarchy collapses to a
three-valued global spectrum.  Let
\(\operatorname{ch}(\mathbf V)\) be the least arity after which \(\cmp_f\)
stabilizes uniformly for all regular languages with syntactic monoid in
\(\mathbf V\).  Our main structural theorem is
\[
 \boxed{\operatorname{ch}(\mathbf V)\in\{1,2,\infty\}},
\]
with the exact characterization
\(\operatorname{ch}(\mathbf V)=\infty \Longleftrightarrow M_{ab}\in\mathbf V \Longleftrightarrow \mathbf V_{ab}\subseteq\mathbf V.\)
Thus finite uniform heights \(3,4,\ldots\) do not occur.

The proof separates the infinite and finite regimes by the same five-element
obstruction.  If \(M_{ab}\in\mathbf V\), every-boundary amplification inside
\(\mathbf V_{ab}\) forces infinite height.  Conversely, if
\(M_{ab}\notin\mathbf V\), a divisor argument shows that
\(\mathbf V\subseteq\mathbf{Com} \qquad\text{or}\qquad \mathbf V\subseteq\mathbf{CR}.\)
Commutative syntax collapses already at arity one.  The main technical input
on the finite side is stronger: for every regular language with completely
regular syntactic monoid, every binary-separating relational morphism is
separating at every finite arity, so
\(\cmp_2=\cmp_3=\cdots\).  Finite bands give a distinct noncommutative
height-one regime, while finite group kernels make the binary bound sharp:
for \(K_\eta=\eta^{-1}(1_G)\),
\(\cmp_1(K_\eta)=1, \qquad \cmp_f(K_\eta)=[G:Z(G)]\quad(f\ge2).\)
The remaining general boundary between heights one and two is open.

The unary optimization also has a direct computational face.  Every nonempty
finite simple graph \(G\) is realized by an explicit length-three language
\(L_G\) with
\(\cmp_1(L_G)=\chi(G),\)
and no additional nonisolated conflict vertices.  Consequently deciding
\(\cmp_1(L)\le3\) is NP-complete for explicitly listed length-three
languages, and the corresponding problem for an explicitly given finite
pointed syntactic monoid is NP-complete as well.  Thus the same invariant
exhibits arbitrary local conflict geometry, strong arity amplification, and a
rigid pseudovariety-level stabilization spectrum.

\subsection{Related work and scope of contribution}

Clark introduced the syntactic concept lattice as a distributional algebra,
related it to syntactic congruence and the universal automaton, and established
its minimal/terminal recognition role~\cite{ClarkSCL}.  Wurm extended the
construction to finite and infinite tuple arities, proved the finite-arity
finiteness criterion, and characterized finiteness of the infinite-tuple SCL
\cite{WurmSCLTuples}; related automata-theoretic semantics appear in
\cite{WurmJLLI2017}, and Kuznetsov develops recent SCL semantics for infinitary
action logic~\cite{Kuznetsov2026}.  Clark--Yoshinaka give a different
SCL-style algebraic extension toward multiple context-free grammars
\cite{ClarkYoshinaka2014}.  To the best of our knowledge, none of these works
minimizes a single finite monoid resource preserving overlap-sensitive principal substitution conflicts
across tuple arities; tuple SCLs themselves are not claimed as new here.

The optimization also sits between several classical minimization traditions.
Syntactic-monoid size is an established algebraic descriptional measure%
~\cite{HolzerKonig2004}; our chain
$\cmp_1(L)\le\cmp_2(L)\le\cdots\le |T_L|$ asks how far that exact resource
can be compressed when only guarded substitution must be preserved.  At a
more abstract level, minimization under a non-transitive compatibility relation
is classical in incompletely specified finite-state machines: Paull--Unger
introduced maximal compatibles and closed covers~\cite{PaullUnger1959}, and
the associated state-reduction problem is NP-complete~\cite{Pfleeger1973}.
We use this only as a structural precedent.  Their compatible cover is not our
relational morphism; the latter must respect monoid multiplication, and our
compatibility constraints are generated by a language rather than supplied as
an incomplete machine specification.

Minimum consistent and separating automata provide a second, algorithmic
precedent for optimizing the size of a finite separator%
~\cite{PittWarmuth1993,ChenFarzanClarkeTsayWang2009,
VazquezPargaGarciaLopez2016,LaumenSnelVaandrager2026}; in compositional
verification, a small separating DFA serves as a reusable finite assumption.
Regular inference also admits direct graph-coloring formulations%
~\cite{CostaFlorencioVerwer2014}.  Our optimization differs in two ways: the
separation specification is generated internally by principal-intent overlap,
and the separator must carry a monoid multiplication reusable across tuple
arities.  The example $L=\{a,a^3\}$ isolates the second difference: its unary
conflict graph is 2-colorable, but $\cmp_1(L)=3$.

Relational morphisms are classical finite-semigroup tools.  Here they arise
canonically because a compression map need not be constant on syntactic
classes: the correct finite semantics is therefore a relation from the
syntactic monoid to the compression codomain rather than, in general, a
quotient map.  This fiber viewpoint is adjacent to the classical use of
relational morphisms in pointlike theory, where simultaneous membership in
fibers records which source elements cannot be separated by a prescribed class
of finite targets~\cite{Almeida1994,RhodesSteinberg2009}.  We are not solving a
pointlike-set problem: our relation is constrained by the pointed unsafe-tuples
of one language, and the objective is the \emph{least cardinality of a single
separating codomain}.  Proposition~\ref{conf:prop:strict-ladder-family} below
also shows that this relational freedom is quantitatively essential: restricting
to syntactic quotients incurs an unbounded additive cost even for unary finite
languages.

The pseudovariety-level invariant $\operatorname{ch}(\mathbf V)$ has a
separate methodological precedent.  Algebraic complexity theory studies
programs uniformly over classes of finite monoids%
~\cite{BarringtonTherien1988}; more recent work identifies tameness phenomena
and resource hierarchies inside particular monoid classes%
~\cite{GrosshansMcKenzieSegoufin2017}.  No technical reduction is claimed: our
resource is codomain cardinality and our parameter is tuple arity.  The
parallel is that a local resource hierarchy is being uniformized over an
algebraically defined source class.

The guarded substitution condition originates in distributional grammatical
inference: see Clark--Eyraud for substitutable CFGs and Yoshinaka for
multidimensional substitutability
\cite{ClarkEyraud2007,Yoshinaka2009,Yoshinaka2011}.  A companion MCFG
manuscript by the author~\cite{KuriyamaMCFG2026} uses the same numerical
minimum, under the notation \(\obs_f\), as a learning-theoretic information
parameter.  Accordingly, the numerical invariant itself is not claimed as
new here.  The new contribution is not the numerical minimum itself but its
intrinsic finite-semigroup structure: the exact relational-morphism
characterization for regular languages, the resulting arity hierarchy and its
amplification phenomena, the completely-regular saturation theorem, and the
pseudovariety-level compression-height trichotomy.  No learning,
reconstruction, or grammar-rank result from the companion manuscript is used
in any proof below.

Finally, the pseudovariety and Rees-matrix background is standard
\cite{Pin1986,Almeida1994,RhodesSteinberg2009,Howie1995,Petrich1969,
Petrich1974}, and \(\operatorname{Synt}(\{ab\})\) occurs among the classical
minimal noncommutative aperiodic generators~\cite{MargolisPin1984}.
Compatible tolerances are classical
\cite{ZelinkaTolerance1975,Pondelicek1978,Loganathan1986,KumaresanTolerance1984}.  In modern
terminology, Barber--Ru\v{s}kuc study reflexive compatible relations as
diagonal subsemigroups and adapt Rees-matrix congruence structure to describe
them~\cite{BarberRuskucDirectSquare,BarberRuskucDiagonal2026}.  Those general results do not by themselves yield the saturation theorem
of Section~\ref{sec:cr-saturation}: our fiber-overlap relation is a
special symmetric diagonal subsemigroup, but the additional \emph{pointed}
binary-separation condition relative to \(P\) forces row/column rigidity and a
unique multiplicative central defect.  The new conclusion is the all-arity
stabilization theorem, not a classification of diagonal subsemigroups.

\subsection{Organization}

Sections~\ref{sec:scl-geometry}--\ref{sec:algebraic-conflict} develop the SCL,
word/context, congruence, and relational-morphism interfaces, ending with the
commutative height-one bound.  Section~\ref{sec:vab-arity} proves
every-boundary amplification, finite-prefix universality, and introduces
compression height.  Section~\ref{sec:cr-saturation} proves the completely
regular arity-two saturation theorem.  Section~\ref{sec:global-arity-trichotomy}
combines these ingredients with the structural divisor argument to obtain the
global trichotomy and records sharp group, band, and benchmark regimes.
Section~\ref{sec:graph-universality} gives graph realization and
NP-completeness, and Section~\ref{sec:discussion} closes with the remaining
problems.

\section{Syntactic Concept Geometry and Finite-Monoid Compression}
\label{sec:scl-geometry}

\paragraph{Notation guide.}
The most frequently used symbols are collected here for reference.
\begin{center}
\small
\begin{tabular}{@{}l p{0.72\linewidth}@{}}
\toprule
Symbol & Meaning\\
\midrule
\(\D_L^{(d)}(\vec x)\) & accepting \(d\)-tuple contexts of \(\vec x\)\\
\(\Obj_d(L),\Om_d(L)\) & principal \(d\)-ary SCL object concepts and their overlap graph\\
\(\Gamma_L,\Gamma_L^{\mathrm{syn}}\) & word-level and syntactic unary conflict graphs\\
\(\mathfrak U_d(L),\mathcal U_d(T,P)\) & semantic and finite syntactic unsafe-pair relations\\
\(\Phi_d(\mathbf t)\) & finite syntactic profile of a \(d\)-tuple\\
\(\operatorname{Safe}_f(h;L)\) & safety of a compression map through arity \(f\)\\
\(\cmp_f(L),\qcmp_f(L)\) & unrestricted and quotient-restricted finite-monoid compression numbers\\
\(\operatorname{ch}(\mathbf V)\) & uniform SCL compression height of a pseudovariety\\
\bottomrule
\end{tabular}
\end{center}

\subsection{The Clark--Wurm finite-tuple tower}

Fix \(L\subseteq\Sigma^*\) and \(d\ge1\).  Following Wurm
\cite{WurmSCLTuples}, the arity-\(d\) object set is \((\Sigma^*)^d\), the
context set is \((\Sigma^*)^{d+1}\), and the incidence relation is
\[
(w_1,\ldots,w_d)\ I_d\ (x_0,\ldots,x_d)
\quad\Longleftrightarrow\quad
x_0w_1x_1\cdots w_dx_d\in L.
\]
We use the equivalent fixed-order hole notation
\(E=x_0\blank_1x_1\blank_2\cdots x_{d-1}\blank_dx_d\), for which
\(E[\vec w]=x_0w_1x_1\cdots w_dx_d\).
Thus a ``tuple context'' always means this order-preserving Clark--Wurm
context; empty tuple components are allowed.  No permutation-enriched context
notion is used anywhere in the paper.
For \(A\subseteq(\Sigma^*)^d\) and \(C\subseteq(\Sigma^*)^{d+1}\), put
\[
A^\uparrow=\{\mathbf x:\mathbf w I_d\mathbf x\text{ for every }\mathbf w\in A\},
\qquad
C^\downarrow=\{\mathbf w:\mathbf w I_d\mathbf x\text{ for every }\mathbf x\in C\}.
\]
The closed object sets \(A^{\uparrow\downarrow}\), equivalently the formal
concepts \((A^{\uparrow\downarrow},A^\uparrow)\), form Wurm's
\(\SCL_d(L)\).  Componentwise concatenation of tuples, followed by closure,
gives its monoid operation; residuals then make it a residuated lattice
\cite{WurmSCLTuples}.

For a tuple \(\mathbf w\), write
\(
 \gamma_d(\mathbf w)
 :=(\{\mathbf w\}^{\uparrow\downarrow},\{\mathbf w\}^\uparrow)
\)
for its \emph{object concept}, and let
\(
 \Obj_d(L):=\{\gamma_d(\mathbf w):\mathbf w\in(\Sigma^*)^d\}
\)
be the principal/object layer.  At \(d=1\), Clark's result
\cite[Proposition~3.19]{ClarkSCL} gives
\(
 \gamma_1(u)=\gamma_1(v)
 \quad\Longleftrightarrow\quad
 u\equiv_L v.
\)
Because singleton concepts multiply by
\(\gamma_1(u)\circ\gamma_1(v)=\gamma_1(uv)\), we obtain the canonical monoid
identification
\begin{proposition}[Unary principal monoid]
\label{scl:prop:unary-principal}
For every language \(L\), the monoid of unary object concepts is canonically
\(
 \Obj_1(L)\cong \Sigma^*/{\equiv_L}.
\)
For regular \(L\), this is the finite syntactic monoid \(T_L\).
\end{proposition}

\subsection{Principal-intent overlap}

For a tuple \(\vec x\in(\Sigma^*)^d\), write
\(
 \D_L^{(d)}(\vec x):=\{(u_0,\ldots,u_d):u_0x_1u_1\cdots x_du_d\in L\}.
\)
This is exactly the intent of the principal concept \(\gamma_d(\vec x)\).
For \(d=1\) we write \(\D_L(x)\).

\begin{definition}[Principal-overlap graph]
For \(d\ge1\), let \(\Om_d(L)\) be the simple graph with vertex set
\(\Obj_d(L)\), where distinct object concepts \(A,B\) are adjacent iff
\(
 \operatorname{Int}(A)\cap\operatorname{Int}(B)\ne\varnothing.
\)
\end{definition}

Thus adjacency means that two distributionally different tuple concepts are
licensed by at least one common accepting context.  A finite monoid
homomorphism \(h:\Sigma^*\to M\) induces the componentwise color
\(h^{(d)}(w_1,\ldots,w_d):=(h(w_1),\ldots,h(w_d))\in M^d\)
on tuples.  The safety condition asks which of these componentwise colors may
be identified without collapsing a principal-overlap conflict.

\begin{definition}[Finite-monoid safety through arity \(f\)]
\label{scl:def:safety}
Let \(h:\Sigma^*\to M\) be a homomorphism into a finite monoid and let
\(f\ge1\).  We write
\(\operatorname{Safe}_f(h;L)\)
if, for every \(1\le d\le f\), no pair of tuples generating distinct
adjacent vertices of \(\Om_d(L)\) receives the same componentwise
\(h\)-type.  Equivalently, for any two distinct adjacent principal concepts,
no choice of one generating tuple for each yields equal componentwise
\(h\)-types.
Synonymously, we say that \(L\) is \((f,h)\)-tuple-substitutable.  These are
not two notions: throughout the paper they are two names for the same safety
predicate, the first emphasizing the SCL geometry and the second the guarded
substitution task.
\end{definition}

Since the intent of \(\gamma_d(\vec x)\) is exactly
\(\D_L^{(d)}(\vec x)\), an edge of \(\Om_d(L)\) is precisely an
overlap-without-equivalence pair: the two tuples share at least one accepting
context but are distinguished by some other context.  Thus
\(\operatorname{Safe}_f(h;L)\) can already be read operationally as saying
that equality of finite \(h\)-type may identify such tuples only when their
complete tuple distributions coincide.  For regular languages this same
witness has congruence and finite relational forms.  We keep the
surface-specific terminology---\emph{safe} for compression maps and
\emph{separating} for congruences and relational morphisms---and record the
exact equivalences in the Safety Interface Theorem,
Theorem~\ref{alg:thm:safety-interface}.

\begin{remark}[Why the principal layer is the substitution layer]
\label{scl:rem:principal-canonical}
The restriction to principal concepts is not an arbitrary truncation of the
SCL.  The objects actually inserted into a tuple context are words and tuples
of words, and their images in the Clark--Wurm lattice are exactly the
principal concepts \(\gamma_d(\vec x)\).  General closed concepts represent
Galois-closed aggregates of tuples rather than additional compositional inputs.
Thus Definition~\ref{scl:def:safety} identifies the precise conflict
geometry forced by guarded tuple substitution.  Our claim of canonicity is
relative to this task: we do not claim that principal-overlap compression is
the unique meaningful compression notion for the full concept lattice.
Table~\ref{tab:scl-dictionary} summarizes the resulting dictionary.
\end{remark}

\begin{table}[t]
\centering
\caption{Dictionary between principal SCL geometry and tuple-substitution language.}
\label{tab:scl-dictionary}
\begin{tabular}{@{}ll@{}}
\toprule
\textbf{SCL language} & \textbf{Word/context language}\\
\midrule
principal object concept \(\gamma_d(\vec x)\) & tuple distribution class of \(\vec x\)\\
intent of \(\gamma_d(\vec x)\) & \(\D_L^{(d)}(\vec x)\)\\
distinct principal concepts & unequal tuple distributions\\
overlapping intents & a shared accepting tuple context\\
edge of \(\Om_d(L)\) & an unsafe tuple pair\\
finite-monoid compression & coordinatewise conflict separation\\
\bottomrule
\end{tabular}
\end{table}

\subsection{Compression number}

Define
\(
 \cmp_f(L)
 :=\min\{|\operatorname{im}(h)|:h\text{ is safe through arity }f\},
\)
with value \(\infty\) if no finite witness exists.  We call this the
\emph{finite-monoid compression number} of the principal SCL geometry through
arity \(f\).

A useful comparison ignores compositional compatibility entirely.

\begin{remark}[Relation to the companion MCFG observation parameter]
\label{rem:obs-cmp-reindex}
The companion MCFG manuscript~\cite{KuriyamaMCFG2026} writes
\(\obs_f(L)\) for the minimum image size of a finite observation making \(L\)
orientation-aware \((f,h)\)-tuple-substitutable.  Its named contexts are
partitioned into permutation sectors, whereas the Clark--Wurm convention
used here fixes the surface order of the tuple coordinates.

The companion condition is sectorwise: substitutability is tested within
each fixed permutation sector, not by comparing contexts belonging to
different sectors.  For every \(\sigma\in S_d\), reindexing
\((x_1,\ldots,x_d) \longmapsto (x_{\sigma(1)},\ldots,x_{\sigma(d)})\)
identifies the \(\sigma\)-sector with the fixed-order contexts used here.
Since the fixed-order safety condition quantifies over all tuples, it applies
equally to every such reindexed pair.  Moreover, applying the same
permutation to both tuples preserves componentwise equality of their
\(h\)-values.

Consequently the safety conditions agree sector by sector, and whenever the
two notations are applied to the same language,
\(\obs_f(L)=\cmp_f(L).\)
We retain the symbol \(\cmp_f\) because the present paper studies this number
as an SCL compression invariant rather than as a learning-side information
parameter.
\end{remark}

\begin{proposition}[Geometric lower bound]
\label{scl:prop:chromatic-lower}
For every language \(L\) and every \(d\ge1\),
\(
 \left\lceil\chi(\Om_d(L))^{1/d}\right\rceil\le \cmp_d(L).
\)
\end{proposition}
\begin{proof}
Let \(h\) be a witness with \(m\) image values.  Assign to each object concept
\(A\in\Obj_d(L)\) any componentwise \(h\)-value realized by a tuple generating
\(A\).  This choice need not be canonical.  Nevertheless, if adjacent concepts
received the same vector in \(M^d\), the chosen representatives would have the
same principal intents as those concepts.  Their tuple distributions would
therefore be distinct but overlapping, so the representatives would form an
unsafe pair collapsed in every coordinate, contradicting safety.  Hence this
is a proper coloring using at most \(m^d\) colors.  Thus \(\chi(\Om_d(L))\le m^d\).  Since \(m\) is a
positive integer, \(\lceil\chi(\Om_d(L))^{1/d}\rceil\le m\).  Minimizing
over witnesses gives the claim.  The same argument shows that if
\(\chi(\Om_d(L))\) is infinite, then no finite witness exists and hence
\(\cmp_d(L)=\infty\).
\end{proof}

\subsection{Finite SCL reduction for regular languages}

Assume \(L\) is regular, let \(\eta:\Sigma^*\twoheadrightarrow T\) be its
syntactic morphism, and let \(P=\eta(L)\).  For \(k\ge1\), write
\(\eta^{(k)}:(\Sigma^*)^k\to T^k\) for the componentwise extension.  Define
the finite formal context
\(
 \mathcal K_d(T,P)=\bigl(T^d,T^{d+1},I_{d,T,P}\bigr),
\)
where
\(
 (t_1,\ldots,t_d)\ I_{d,T,P}\ (q_0,\ldots,q_d)
 \quad\Longleftrightarrow\quad
 q_0t_1q_1\cdots t_dq_d\in P.
\)
The maps \(\eta^{(d)}\) on objects and \(\eta^{(d+1)}\) on contexts are surjective
and preserve incidence in both directions.  Hence duplicate rows and columns
are the only information lost.

\begin{proposition}[Finite concept reduction]
\label{scl:prop:finite-reduction}
For every regular \(L\) and finite \(d\), \(\SCL_d(L)\) is isomorphic to the
concept lattice of \(\mathcal K_d(T_L,P)\).  In particular the principal
intent of \(\mathbf t=(t_1,\ldots,t_d)\in T^d\) is
\(
 \Phi_d(\mathbf t)
 =\{(q_0,\ldots,q_d)\in T^{d+1}:q_0t_1q_1\cdots t_dq_d\in P\}.
\)
\end{proposition}
\begin{proof}
For concrete tuples \(\mathbf w\) and contexts \(\mathbf x\), incidence is
equivalent to incidence of \(\eta^{(d)}(\mathbf w)\) and
\(\eta^{(d+1)}(\mathbf x)\) in \(\mathcal K_d(T,P)\).  Both maps are
surjective.  If \(\uparrow_T,\downarrow_T\) denote the Galois operators of the
finite context, then for all object sets \(A\subseteq(\Sigma^*)^d\) and context
sets \(C\subseteq(\Sigma^*)^{d+1}\),
\(A^\uparrow = (\eta^{(d+1)})^{-1}\!\bigl((\eta^{(d)}(A))^{\uparrow_T}\bigr), \qquad C^\downarrow = (\eta^{(d)})^{-1}\!\bigl((\eta^{(d+1)}(C))^{\downarrow_T}\bigr).\)
Hence every closed extent and intent of the infinite context is saturated with
respect to the relevant syntactic map and is the inverse image of a unique
closed extent or intent in the finite context; surjectivity gives the converse.
The two concept lattices are therefore isomorphic.
\end{proof}

The remaining sections turn this SCL formulation into congruence and
relational-morphism optimization and then classify its arity behavior.

\section{Word-Level Conflict and Compositional Compression}
\label{sec:conflict}

All monoids in this paper have an identity, and all homomorphisms preserve it.
Fix a finite alphabet \(\Sigma\).  We retain the fixed-order Clark--Wurm tuple
contexts and filling notation \(E[\vec x]\) from
Section~\ref{sec:scl-geometry}.  The tuple distribution
\(\D_L^{(d)}(\vec x)\) is therefore equivalently the set of tuple contexts
\(E\) with \(E[\vec x]\in L\).

No new safety predicate is introduced at the word level.  We use the synonym
``\(L\) is \((f,h)\)-tuple-substitutable'' from
Definition~\ref{scl:def:safety}.  The point of the present section is to expose
the resulting semantic conflict clauses explicitly; Theorem~\ref{alg:thm:safety-interface}
will later identify the same predicate with its finite syntactic and relational
forms.

\begin{proposition}[Recognition-to-compression relaxation]
\label{conf:prop:recognition-relaxation}
Let \(h:\Sigma^*\to M\) be a monoid morphism recognizing \(L\), so that
\(L=h^{-1}(P)\) for some \(P\subseteq M\). Then \(h\) witnesses
\((f,h)\)-tuple substitutability for every finite \(f\). Consequently, if \(L\)
is regular, then
\(
\cmp_f(L)\le |T_L|
\qquad(f\ge1),
\)
where \(T_L\) is the syntactic monoid of \(L\).
\end{proposition}

\begin{proof}
Suppose \(h^{(d)}(\vec x)=h^{(d)}(\vec y)\) for some finite \(d\). For every
tuple context
\(E=u_0\blank_1u_1\cdots\blank_du_d\),
compatibility of \(h\) with concatenation gives
\(
h(E[\vec x])=h(E[\vec y]).
\)
Since membership in \(L\) depends only on the \(h\)-value, we have
\(E[\vec x]\in L\) if and only if \(E[\vec y]\in L\). Hence the two tuple
distributions are equal, so \(h\) is safe at every finite arity. For regular \(L\), apply the statement to the syntactic morphism
\(\eta_L:\Sigma^*\twoheadrightarrow T_L\).
\end{proof}

\begin{remark}[Exact recognition versus sufficient compression]
Proposition~\ref{conf:prop:recognition-relaxation} makes the relaxation
literal: exact recognizers form a restricted feasible class for the
compression problem. The invariant \(\cmp_f\) therefore measures how much
finite compositional information can be forgotten when the task is weakened
from exact recognition to contextual-substitution safety through arity \(f\).
The examples below show that this relaxation can remain arbitrarily large even
when all finite arities are required, while in other families the hierarchy
stabilizes already at arity two.  For kernel languages of centerless finite
groups, the binary level even recovers the full syntactic-monoid size.
\end{remark}

\begin{definition}[Contextual conflict / unsafe tuple pair]
For \(d\ge1\), put
\(\mathfrak U_d(L) :=\left\{ (\vec x,\vec y): \D_L^{(d)}(\vec x)\cap\D_L^{(d)}(\vec y)\neq\varnothing, \ \D_L^{(d)}(\vec x)\neq\D_L^{(d)}(\vec y) \right\}.\)
Its elements are the \emph{contextual conflicts}, or \emph{unsafe
arity-\(d\) tuple pairs}, of \(L\).
\end{definition}

The two requirements in this definition have distinct roles. If the
complete distributions are equal, substitution is semantically safe and no
separation is needed. If the distributions are disjoint, there is no
successful context in which the two tuples are simultaneously licensed, so
the guarded substitution principle never compares them. The unsafe pairs are
therefore exactly the overlap-without-equivalence configurations that must be
resolved. Equivalently, \(h\) witnesses \(\cmp_f\) precisely when every
\((\vec x,\vec y)\in\mathfrak U_d(L)\), \(d\le f\), satisfies \(h(x_1)\neq h(y_1)\ \vee\cdots\vee\ h(x_d)\neq h(y_d).\) 

\subsection{The arity-one conflict graph}

\begin{definition}[Contextual-conflict graph]
For a language \(L\subseteq\Sigma^*\), let \(\Gamma_L\) be the simple
graph with vertex set \(\Sigma^*\) and an edge \(\{x,y\}\) for each
distinct unsafe unary pair, equivalently when
\(
\D_L(x)\cap\D_L(y)\neq\varnothing
\quad\text{and}\quad
\D_L(x)\neq\D_L(y).
\)
\end{definition}

\begin{theorem}[Compositional-coloring characterization]
\label{conf:thm:coloring}
For every language \(L\),
\[
\cmp_1(L)
=
\min\left\{
|\operatorname{im}(h)|:
\begin{array}{l}
h:\Sigma^*\to M\text{ is a finite-monoid}\\
\text{homomorphism and a proper coloring of }\Gamma_L
\end{array}
\right\}.
\]
In particular, \(\chi(\Gamma_L)\le\cmp_1(L).\) \end{theorem}

\begin{proof}
At arity one, a homomorphism witnesses tuple substitutability exactly when no
unsafe pair receives one value. This is exactly the proper-coloring condition
for \(\Gamma_L\). Forgetting the homomorphism constraint leaves an ordinary
proper coloring, giving the inequality.
\end{proof}

Thus ordinary coloring measures the conflict geometry alone, while
\(\cmp_1\) measures the cost of coloring that geometry in a way compatible
with concatenation. The difference between the two quantities is the price of
compositionality.

\begin{remark}[A strict price of compositionality]
\label{conf:rem:strict-price}
Let \(L=\{a,a^3\}\subseteq\{a\}^*\).  Its only nontrivial unary conflict
edges are \(\{\varepsilon,a^2\}\) and \(\{a,a^3\}\), so
\(\chi(\Gamma_L)=2\).  Nevertheless \(\cmp_1(L)=3\).

A one-element image cannot be a witness, since it collapses both conflict
edges.  Thus it remains to exclude witnesses with exactly two image values.
Such a homomorphism is determined by \(x=h(a)\), and necessarily
\(x\ne1\).  The two possibilities for a two-element monoid generated by
\(x\) are the group case \(x^2=1\) and the idempotent case \(x^2=x\).
The former identifies \(\varepsilon\) with \(a^2\), while the latter gives
\(x^3=x\) and identifies \(a\) with \(a^3\).  Thus no two-element
compression map separates both conflict edges.  Conversely, mapping \(a\) to a
generator of the cyclic group \(C_3\) separates both edges.  Hence
\(\chi(\Gamma_L)=2<3=\cmp_1(L).\)
This is the smallest kind of phenomenon that ordinary graph coloring misses:
the extra cost comes solely from requiring the coloring to be compositional.
\end{remark}

\subsection{Regular languages: from coloring to syntactic semantics}

For regular \(L\), let \(\eta_L:\Sigma^*\twoheadrightarrow T_L\) be the syntactic morphism. A compression map need not recognize \(L\) and need
not be constant on syntactic classes; it is required only to separate
contextual conflicts.

Because unary context distributions are constant on syntactic classes, the
infinite graph \(\Gamma_L\) has a finite semantic core. Define the
\emph{syntactic conflict graph} \(\Gamma_L^{\mathrm{syn}}\) on \(T_L\) by
joining distinct elements \(s,t\) whenever (equivalently, for any
representatives \(x,y\) with \(\eta_L(x)=s\) and \(\eta_L(y)=t\)) the
words \(x,y\) form an unsafe unary pair.

\begin{proposition}[Finite conflict core of a regular language]
\label{conf:prop:syntactic-conflict-graph}
For regular \(L\), the graph \(\Gamma_L\) is obtained from
\(\Gamma_L^{\mathrm{syn}}\) by replacing each syntactic vertex by an
independent set of word representatives. In particular, \(\chi(\Gamma_L)=\chi(\Gamma_L^{\mathrm{syn}}).\) \end{proposition}

\begin{proof}
Words in one syntactic class have identical two-sided context distributions,
so they are never adjacent. For two distinct syntactic classes, both
intersection and inequality of the context distributions depend only on the
classes, hence either every pair of representatives is adjacent or none is.
Thus \(\Gamma_L\) is a blow-up of the finite syntactic conflict graph, and
blowing up vertices into nonempty independent sets preserves chromatic number.
\end{proof}

\begin{corollary}[Unary SCL overlap equals syntactic conflict]
\label{conf:cor:omega-syntactic}
For regular \(L\), under the canonical identification
\(\Obj_1(L)\cong T_L\) of Proposition~\ref{scl:prop:unary-principal},
\(
 \Om_1(L)=\Gamma_L^{\mathrm{syn}}.
\)
Consequently
\(
 \chi(\Gamma_L)=\chi(\Gamma_L^{\mathrm{syn}})=\chi(\Om_1(L)).
\)
\end{corollary}
\begin{proof}
The intent of the unary principal concept represented by a word \(x\) is
\(\D_L(x)\).  Definition~\ref{scl:def:safety} identifies adjacency in
\(\Om_1(L)\) exactly with the unsafe relation between the corresponding
syntactic classes.  The remaining chromatic equality is
Proposition~\ref{conf:prop:syntactic-conflict-graph}.
\end{proof}

\begin{definition}[Quotient-restricted compression size]
For regular \(L\), let \(\qcmp_f(L)\) be the minimum size of a quotient
monoid \(Q\) of \(T_L\) such that \(\pi\circ\eta_L:\Sigma^*\to Q\) witnesses
\((f,\pi\circ\eta_L)\)-tuple substitutability.
\end{definition}

\begin{proposition}[The conflict--syntax complexity ladder]
\label{conf:prop:ladder}
For every regular language \(L\), \(\chi(\Gamma_L) \le \cmp_1(L) \le \qcmp_1(L) \le |T_L|.\) \end{proposition}

\begin{proof}
The first inequality is Theorem~\ref{conf:thm:coloring}. The second follows
because quotient compressions are a restricted class of finite compression maps.
For the third, take the identity quotient of \(T_L\). If two words have the
same syntactic value, substitution inside every two-sided context preserves
membership, so the syntactic morphism itself is a witness.
\end{proof}

The ladder separates four resources: unconstrained conflict coloring,
compositional conflict coloring, quotient-based abstraction of syntactic
semantics, and the full syntactic monoid.  All four inequalities can be strict,
and the middle gap can be arbitrarily large.

\begin{proposition}[A strict ladder with unbounded relational--quotient gap]
\label{conf:prop:strict-ladder-family}
For $n\ge2$, let
\(K_n=\{a^n,a^{n+2}\}\subseteq\{a\}^*.\)
Then for every $f\ge1$,
\(\chi(\Gamma_{K_n})=2 <3=\cmp_f(K_n) <n+2=\qcmp_f(K_n) <n+4=|T_{K_n}|.\)
In particular, $\qcmp_f(K_n)-\cmp_f(K_n)=n-1$, so restricting finite
compression to quotients of the syntactic monoid has unbounded additive cost.
\end{proposition}

\begin{proof}
Write $t_i$ for the syntactic class of $a^i$, $0\le i\le n+2$, and $\bot$
for the common class of all $a^j$ with $j\ge n+3$.  These are precisely the
syntactic classes: distinct lengths at most $n+2$ are separated by a context
that completes one of them to length $n$ or $n+2$, while every longer word has
empty distribution.  Thus $|T_{K_n}|=n+4$, with truncated addition
$t_it_j=t_{i+j}$ when $i+j\le n+2$ and product $\bot$ otherwise.

The nontrivial unary conflict edges are exactly
\(\{t_i,t_{i+2}\},\qquad 0\le i\le n.\)
Hence the finite conflict core is the disjoint union of the two paths on even
and odd indices, so $\chi(\Gamma_{K_n})=2$.

Mapping $a$ to a generator $g$ of the cyclic group $C_3$ separates every such
edge, since $g^i\ne g^{i+2}$.  Therefore $\cmp_1(K_n)\le3$.

Since the conflict graph has an edge, a one-element image cannot be
unary-safe.  No two-element image works either: if $x$ is the image of $a$
and the generated image has two elements, then either $x^2=1$, which
identifies every pair $a^i,a^{i+2}$, or $x^2=x$, which identifies
$a$ and $a^3$.  Thus $\cmp_1(K_n)=3$.  Since
$T_{K_n}$ is commutative, the commutative-collapse theorem proved below,
Theorem~\ref{group:thm:commutative-collapse}, gives
$\cmp_f(K_n)=3$ for every $f\ge1$.

It remains to compute the quotient-restricted value.  Let
$\pi:T_{K_n}\twoheadrightarrow Q$ be any unary-safe quotient, put
$y=\pi(t_1)$ and $z=\pi(\bot)$, and let $r$ be least with $y^r=z$.  Such an
$r$ exists because $y^{n+3}=z$.  If $r\le n$, then
$\pi(t_r)=z=\pi(t_{r+2})$, collapsing a conflict edge.  Hence $r\ge n+1$.
Moreover $1,y,\ldots,y^{r-1},z$ are pairwise distinct: if
$y^i=y^j$ for $0\le i<j<r$, multiplication by $y^{r-j}$ would give
$y^{i+r-j}=z$ with exponent $i+r-j<r$, contradicting minimality of $r$.
Thus $|Q|\ge r+1\ge n+2$.

Conversely, identify the tail
$\{t_{n+1},t_{n+2},\bot\}$ and keep
$t_0,\ldots,t_n$ distinct.  This is a quotient monoid of size $n+2$, and it
collapses no conflict edge.  Hence $\qcmp_1(K_n)=n+2$.  Again commutative
collapse shows that the same quotient is safe through every finite arity, so
$\qcmp_f(K_n)=n+2$ for all $f\ge1$.
\end{proof}

\begin{remark}[Why the optimal compression is genuinely relational on syntax]
\label{conf:rem:genuinely-relational}
For the optimal map $h(a)=g\in C_3$ above, the canonical relation on the
syntactic monoid satisfies
\(\rho_h(t_i)=\{g^i\}\quad(0\le i\le n+2), \qquad \rho_h(\bot)=C_3.\)
Thus one syntactic element has three possible compression values.  The finite
compression retains residue information inside the dead syntactic class, and
that information disappears under every functional quotient of $T_{K_n}$.
This gives a concrete reason, independent of the abstract selector theorem,
that relational rather than quotient semantics is the correct finite
interface.
\end{remark}

The graph-realization theorem in Section~\ref{sec:graph-universality} later
shows that the first two levels can also coincide on arbitrary finite conflict
geometries.  For the structural theory below, the important next step is higher
arity.

\subsection{Higher arity as a compositional separation system}

For \(d>1\), an unsafe pair is no longer a single graph edge between word
colors. It imposes the disjunctive clause
\(h(x_1)\neq h(y_1) \ \vee\ \cdots\ \vee\ h(x_d)\neq h(y_d).\)
Accordingly, \(\cmp_f(L)\) is the minimum size of a finite monoid whose
canonical coloring of \(\Sigma^*\) satisfies every language-induced
separation clause of width at most \(f\). The graph case is the width-one
shadow of this constraint system. The finite-index congruence theorem below
makes the formulation exact algebraically.

\section{Algebraic Characterization of Conflict Resolution}
\label{sec:algebraic-conflict}

The conflict clauses of Section~\ref{sec:conflict} admit an exact algebraic form. On the free monoid, a finite compression map is the same thing as a finite-index congruence whose classes color all conflicts compositionally. For regular languages, the same optimization can be pushed to the syntactic monoid, but generally only through a relational morphism rather than a functional quotient.

The free-monoid congruence formulation is elementary but useful: it records
the optimization problem in an intrinsic algebraic form. The substantive
regular-language step is the subsequent passage from the free monoid to
relational morphisms of the finite syntactic pointed monoid.

\subsection{Free-monoid separation clauses}

The first algebraic interface is the kernel congruence of a compression map.
We introduce only the feasibility notion here; the optimization theorem is
stated after the Safety Interface Theorem so that it becomes an immediate
consequence of that common interface.

\begin{definition}[Separating congruence]
A finite-index two-sided monoid congruence
\(
\theta\ \subseteq\Sigma^*\times\Sigma^*
\)
is \emph{\(f\)-separating for \(L\)} if, for every
\(
(\vec x,\vec y)\in\mathfrak U_d(L),
\qquad
1\le d\le f,
\)
there exists a coordinate \(i\in\{1,\ldots,d\}\) such that \(x_i\not\equiv_\theta y_i.\) \end{definition}

\subsection{Finite syntactic profiles for regular languages}

Assume from now on that \(L\) is regular.
Let \(\eta:\Sigma^*\twoheadrightarrow T\) be its syntactic morphism and let
\(P:=\eta(L)\subseteq T\). Thus \(L=\eta^{-1}(P)\).  We call \((T,P)\) the
pointed syntactic monoid of \(L\), following the classical pointed-monoid
(or \(P\)-monoid) framework~\cite{Sakarovitch1979,Pin1981}.

Recall from Proposition~\ref{scl:prop:finite-reduction} the finite tuple-context
profile
\[
\Phi_d(\mathbf t)
=\left\{(q_0,\ldots,q_d)\in T^{d+1}:q_0t_1q_1\cdots t_dq_d\in P\right\}
\qquad(\mathbf t\in T^d).
\]

If
\(E=u_0\blank_1u_1\cdots\blank_du_d,\)
write
\(\eta_{\mathrm{ctx}}^{(d)}(E) := (\eta(u_0),\ldots,\eta(u_d))\in T^{d+1}.\)
This notation is reserved for the finite syntactic profile of a tuple context.
\begin{lemma}[Finite-profile transport]
\label{alg:lem:profile}
Let \(\vec x,\vec y\in(\Sigma^*)^d\).  Set
\(\mathbf s:=\eta^{(d)}(\vec x)\) and \(\mathbf t:=\eta^{(d)}(\vec y)\).
Then:
\begin{enumerate}[label=(\roman*),leftmargin=*]
\item for every tuple context \(E\),
\(E\in\D_L^{(d)}(\vec x) \iff \eta_{\mathrm{ctx}}^{(d)}(E)\in\Phi_d(\mathbf s);\)
\item
\(
\D_L^{(d)}(\vec x)=\D_L^{(d)}(\vec y)
\iff
\Phi_d(\mathbf s)=\Phi_d(\mathbf t);
\)
\item
\(
\D_L^{(d)}(\vec x)\cap\D_L^{(d)}(\vec y)\neq\varnothing
\iff
\Phi_d(\mathbf s)\cap\Phi_d(\mathbf t)\neq\varnothing.
\)
\end{enumerate}
\end{lemma}

\begin{proof}
For (i), if
\(E=u_0\blank_1u_1\cdots\blank_du_d,\)
then \(\eta(E[\vec x]) = \eta(u_0)s_1\eta(u_1)\cdots s_d\eta(u_d).\) Since \(L=\eta^{-1}(P)\), this word belongs to \(L\) exactly when
\(\eta_{\mathrm{ctx}}^{(d)}(E)\in\Phi_d(\mathbf s)\).

For (ii), the forward implication follows from (i) once we observe that every
abstract profile
\((q_0,\ldots,q_d)\in T^{d+1}\)
has a concrete context representative: because \(\eta\) is surjective, choose
\(u_j\in\Sigma^*\) with \(\eta(u_j)=q_j\), and form
\(u_0\blank_1u_1\cdots\blank_du_d.\)
Thus a profile belonging to exactly one of
\(\Phi_d(\mathbf s)\) and \(\Phi_d(\mathbf t)\) yields a concrete context
belonging to exactly one of the two tuple distributions.
The reverse implication is immediate from (i).

For (iii), a common concrete accepting context gives a common profile by (i).
Conversely, if an abstract profile belongs to both finite profile sets, choose
concrete representatives \(u_j\) as above. The resulting single named
context accepts both tuples by (i).
\end{proof}

\begin{definition}[Unsafe syntactic tuple pair]
For \(1\le d\le f\), define
\[
\mathcal U_d(T,P)
:=
\left\{
(\mathbf s,\mathbf t)\in T^d\times T^d:
\begin{array}{l}
\Phi_d(\mathbf s)\cap\Phi_d(\mathbf t)\neq\varnothing,\\
\Phi_d(\mathbf s)\neq\Phi_d(\mathbf t)
\end{array}
\right\}.
\]
\end{definition}

By Lemma~\ref{alg:lem:profile}, this is exactly the finite syntactic image of the
semantic unsafe-pair relation.

We use \(\mathfrak U_d(L)\) for semantic unsafe pairs of word tuples and
\(\mathcal U_d(T,P)\) for their finite syntactic-profile counterpart.

\subsection{Canonical relational morphism of a compression map}

\begin{definition}[Monoid relational morphism]
Let \(T,M\) be monoids.  We use the identity-preserving monoid convention: a
\emph{monoid relational morphism} \(\rho:T\relto M\) assigns to every \(t\in T\)
a nonempty set \(\rho(t)\subseteq M\) such that \(1_M\in\rho(1_T)\) and
\(
\rho(s)\rho(t)\subseteq\rho(st)
\qquad(s,t\in T).
\)
Equivalently, its graph
\(
\operatorname{Gr}(\rho)
:=
\{(t,m):m\in\rho(t)\}
\)
is a submonoid of \(T\times M\) whose first projection is surjective.  We do
not require the second projection to be surjective; unused codomain values can
always be deleted, as done explicitly in the proof of
Theorem~\ref{alg:thm:relational-characterization}.
\end{definition}

\begin{definition}[Canonical relational morphism induced by a compression map]
Let \(h:\Sigma^*\to M\) be a finite compression map.
Define \(\rho_h:T\relto M\) by \(\rho_h(t) := \{h(w):\eta(w)=t\}.\) \end{definition}

\begin{lemma}[Canonical relation is a relational morphism]
\label{alg:lem:canonical-rel}
The map \(\rho_h:T\relto M\) is a monoid relational morphism.
Its graph is exactly
\(
\operatorname{im}(\eta,h)
=
\{(\eta(w),h(w)):w\in\Sigma^*\}
\subseteq T\times M.
\)
\end{lemma}

\begin{proof}
Every \(t\in T\) has a preimage under the surjective syntactic morphism, so
\(\rho_h(t)\neq\varnothing\).
Moreover
\(1_M=h(\varepsilon)\in\rho_h(1_T).\)
If
\(m\in\rho_h(s), \qquad n\in\rho_h(t),\)
choose words \(u,v\) such that
\(\eta(u)=s,\quad h(u)=m, \qquad \eta(v)=t,\quad h(v)=n.\)
Then \(\eta(uv)=st, \qquad h(uv)=mn,\) so \(mn\in\rho_h(st).\) The graph identity follows directly from the definition.
\end{proof}

\begin{definition}[Unsafe-separating relational morphism]
\label{alg:def:rel-separating}
A relational morphism \(\rho:T\relto M\) is \emph{\(f\)-separating for \((T,P)\)} if, for every
\(
(\mathbf s,\mathbf t)\in\mathcal U_d(T,P),
\qquad
1\le d\le f,
\)
there exists a coordinate \(i\) such that \(\rho(s_i)\cap\rho(t_i)=\varnothing.\) \end{definition}

\begin{definition}[Safe compatible tolerance]
\label{alg:def:safe-tolerance}
Let \((T,P)\) be a finite pointed monoid.  A relation
\(\tau\subseteq T\times T\) is a \emph{compatible tolerance} if it is
reflexive, symmetric, and
\(
x\tau y,\ u\tau v \Longrightarrow xu\tau yv.
\)
It is \emph{\(f\)-safe for \((T,P)\)} if every
\((\mathbf s,\mathbf t)\in\mathcal U_d(T,P)\), \(d\le f\), is separated in
some coordinate: \(s_i\not\tau t_i\) for some \(i\).
\end{definition}

\begin{lemma}[Relational-overlap interface]
\label{alg:lem:relational-overlap}
Every relational morphism \(\rho:T\relto M\) induces the compatible tolerance
\[
 x\tau_\rho y
 \quad\Longleftrightarrow\quad
 \rho(x)\cap\rho(y)\ne\varnothing.
\tag{T}
\]
Moreover,
\[
\rho\text{ is \(f\)-separating for \((T,P)\)}
\quad\Longleftrightarrow\quad
\tau_\rho\text{ is \(f\)-safe for \((T,P)\)}.
\tag{RT}
\]
\end{lemma}

\begin{proof}
Reflexivity and symmetry follow from the nonempty fibers of \(\rho\).  If
\(x\tau_\rho y\) and \(u\tau_\rho v\), choose
\(m\in\rho(x)\cap\rho(y)\) and \(n\in\rho(u)\cap\rho(v)\).  Then
\(mn\in\rho(xu)\cap\rho(yv)\), proving compatibility.  The equivalence
\textup{(RT)} is immediate from the two definitions: an unsafe tuple pair is
unseparated by \(\rho\) exactly when all of its coordinate pairs are
\(\tau_\rho\)-related.
\end{proof}

\begin{theorem}[Safety Interface Theorem]
\label{alg:thm:safety-interface}
Let \(L\subseteq\Sigma^*\) be any language and let
\(h:\Sigma^*\to M\) be a homomorphism into a finite monoid.  For every
\(f\ge1\), conditions \textup{(i)}--\textup{(iv)} below are equivalent.  If
\(L\) is regular, write
\(\eta:\Sigma^*\twoheadrightarrow T\) for its syntactic morphism and
\(P=\eta(L)\); then they are also equivalent to \textup{(v)}.
\begin{enumerate}[label=(\roman*),leftmargin=*]
\item \textbf{Principal-SCL interface:}
\(\operatorname{Safe}_f(h;L)\) holds; equivalently, for every \(d\le f\),
no pair of generating tuples for two distinct adjacent principal concepts in
\(\Om_d(L)\) has the same componentwise \(h\)-type.  No map from principal
concepts to \(M^d\) is asserted or required.
\item \textbf{Guarded word/context interface:} for every \(1\le d\le f\),
\(h^{(d)}(\vec x)=h^{(d)}(\vec y),\qquad \D_L^{(d)}(\vec x)\cap\D_L^{(d)}(\vec y)\ne\varnothing \quad\Longrightarrow\quad \D_L^{(d)}(\vec x)=\D_L^{(d)}(\vec y).\)
Equivalently, equal finite type plus one shared accepting context forces
interchangeability in every tuple context.
\item \textbf{Semantic-clause interface:} every unsafe semantic pair
\((\vec x,\vec y)\in\mathfrak U_d(L)\), \(d\le f\), is separated in at least
one coordinate:
\(h(x_i)\ne h(y_i)\qquad\text{for some }i.\)
\item \textbf{Congruence interface:} the kernel congruence \(\ker h\) is
\(f\)-separating for \(L\).
\item \textbf{Finite syntactic/relational interface:} the canonical
relational morphism \(\rho_h:T\relto M\) is \(f\)-separating for the pointed
syntactic monoid \((T,P)\).
\end{enumerate}
Thus the SCL, word/context, congruence, and canonical relational-morphism
formulations are witness-level representations of one safety predicate, not
successive relaxations or additional invariants.
\end{theorem}

\begin{proof}
The equivalence of \textup{(i)} and \textup{(ii)} is the principal-intent
dictionary.  The intent of \(\gamma_d(\vec x)\) is
\(\D_L^{(d)}(\vec x)\); hence two principal concepts are distinct and adjacent
exactly when their distributions overlap and are unequal.  Choosing generating tuples for such adjacent concepts and giving those
representatives equal componentwise \(h\)-type is therefore exactly the negation
of the implication in \textup{(ii)}.

The equivalence of \textup{(ii)} and \textup{(iii)} is the same statement in
clause form: an unsafe pair is, by definition, a pair with a shared accepting
context and unequal complete distributions.  Thus safety says precisely that
no unsafe pair has equal \(h\)-values in every coordinate.  Since
\(x\equiv_{\ker h}y\) if and only if \(h(x)=h(y)\), this is also exactly the
\(f\)-separating condition for \(\ker h\), proving
\textup{(iii)}\(\Longleftrightarrow\)\textup{(iv)}.

It remains to compare the semantic and finite syntactic interfaces.  Suppose
\textup{(iii)} holds and that some
\((\mathbf s,\mathbf t)\in\mathcal U_d(T,P)\), \(d\le f\), is not separated
by \(\rho_h\).  For each coordinate choose
\(m_i\in\rho_h(s_i)\cap\rho_h(t_i)\)
and words \(x_i,y_i\) with
\(\eta(x_i)=s_i,\quad \eta(y_i)=t_i,\quad h(x_i)=h(y_i)=m_i.\)
By Lemma~\ref{alg:lem:profile}, the word tuples \(\vec x,\vec y\) form an
unsafe semantic pair, but no coordinate is separated by \(h\), contradicting
\textup{(iii)}.  Hence \(\rho_h\) is \(f\)-separating.

Conversely, suppose \textup{(v)} holds and let
\((\vec x,\vec y)\in\mathfrak U_d(L)\), \(d\le f\).  Put
\(\mathbf s=\eta^{(d)}(\vec x)\) and
\(\mathbf t=\eta^{(d)}(\vec y)\).  Lemma~\ref{alg:lem:profile} gives
\((\mathbf s,\mathbf t)\in\mathcal U_d(T,P)\).  If
\(h(x_i)=h(y_i)\) for every \(i\), then this common value belongs to
\(\rho_h(s_i)\cap\rho_h(t_i)\) in every coordinate, contradicting the
\(f\)-separation of \(\rho_h\).  Thus some coordinate is separated and
\textup{(iii)} follows.
\end{proof}

\begin{table}[t]
\centering
\caption{Typed safety dictionary.  The first four rows are equivalent for a
fixed compression witness \(h\); the last two are equivalent for an arbitrary
relational witness \(\rho\).}
\label{tab:typed-safety-dictionary}
\small
\begin{tabular}{@{}>{\raggedright\arraybackslash}p{0.27\linewidth}
                  >{\raggedright\arraybackslash}p{0.28\linewidth}
                  >{\raggedright\arraybackslash}p{0.34\linewidth}@{}}
\toprule
\textbf{Carrier} & \textbf{Witness} & \textbf{Feasibility predicate}\\
\midrule
principal SCL / words & \(h:\Sigma^*\to M\) & \(\operatorname{Safe}_f(h;L)\)\\
word conflicts & \(h:\Sigma^*\to M\) & coordinatewise separation of \(\mathfrak U_d(L)\)\\
free monoid & \(\ker h\) & \(f\)-separating for \(L\)\\
syntactic monoid & canonical \(\rho_h:T\relto M\) & \(f\)-separating for \((T,P)\)\\
\midrule
syntactic monoid & arbitrary \(\rho:T\relto M\) & \(f\)-separating for \((T,P)\)\\
syntactic monoid & overlap \(\tau_\rho\) & \(f\)-safe for \((T,P)\)\\
\bottomrule
\end{tabular}
\end{table}

\begin{remark}[Witness-level scope of the interface]
\label{alg:rem:witness-scope}
The relational statement in Theorem~\ref{alg:thm:safety-interface} concerns
only the canonical relation \(\rho_h\) induced by the fixed compression map
\(h\).  It does not say that an arbitrary relational morphism is itself the
canonical relation of some homomorphism.  Lemma~\ref{alg:lem:selector} below
is the separate freeness argument that extracts from an arbitrary relational
morphism a functional witness whose canonical fibers are contained in the
original fibers.  This distinction is what upgrades the witness-level
interface to the least-codomain characterization of
Theorem~\ref{alg:thm:relational-characterization}.
\end{remark}

\subsection{Free-monoid congruence characterization}

The witness-level equivalence above immediately yields the free-monoid
optimization statement.

\begin{theorem}[Finite-index congruence characterization]
\label{alg:thm:free-congruence}
For every language \(L\subseteq\Sigma^*\) and every \(f\ge1\),
\[
\cmp_f(L)
=
\min\left\{
[\Sigma^*:\theta]:
\begin{array}{l}
\theta\text{ is a finite-index monoid congruence}\\
\text{that is \(f\)-separating for }L
\end{array}
\right\}.
\]
Here and below, the minimum is understood to be \(\infty\) if no such finite-index
congruence exists.
\end{theorem}

\begin{proof}
Let \(h:\Sigma^*\to M\) be a finite witnessing compression map and replace
\(M\) by \(\operatorname{im}(h)\).  By the Safety Interface Theorem,
\(\ker h\) is \(f\)-separating, while the first isomorphism theorem gives
\([\Sigma^*:\ker h]=|\operatorname{im}(h)|.\)
Hence the minimum on the right is at most \(\cmp_f(L)\).

Conversely, let \(\theta\) be a finite-index \(f\)-separating congruence and
let \(q_\theta:\Sigma^*\to\Sigma^*/\theta\) be the quotient morphism.  Its
kernel is exactly \(\theta\), so the congruence interface
\textup{(iv)}\(\Rightarrow\)\textup{(i)} of
Theorem~\ref{alg:thm:safety-interface} shows that \(q_\theta\) is safe through
arity \(f\).  Therefore
\(\cmp_f(L)\le |\Sigma^*/\theta|=[\Sigma^*:\theta].\)
Taking the minimum proves the equality.
\end{proof}

Recognizing congruences form a restricted feasible class: if a finite-index
monoid congruence saturates \(L\), its quotient morphism recognizes \(L\) and
hence is safe at every finite arity by
Proposition~\ref{conf:prop:recognition-relaxation}.  Therefore, with
\(\mathsf{Sep}_f(L)\) denoting the finite-index \(f\)-separating congruences,
\(\mathsf{Rec}(L)\subseteq\bigcap_{f\ge1}\mathsf{Sep}_f(L)\), while
\(\mathsf{Sep}_{f+1}(L)\subseteq\mathsf{Sep}_f(L)\).
For regular \(L\), minimizing congruence index in
Theorem~\ref{alg:thm:free-congruence} therefore gives the nested relaxation
chain
\(\cmp_1(L)\le\cmp_2(L)\le\cdots\le |T_L|.\)
Thus increasing arity imposes successively stronger separation constraints,
without in general forcing exact syntactic recognition.

\begin{remark}[Clause form]
For an unsafe arity-\(d\) pair \((\vec x,\vec y)\), write
\(\vec x=(x_1,\ldots,x_d)\) and \(\vec y=(y_1,\ldots,y_d)\).  The separation
requirement is the clause
\(\bigvee_{i=1}^d(x_i\not\equiv_\theta y_i)\).
Thus \(\cmp_f(L)\) is the minimum index of a finite monoid congruence
satisfying all unsafe separation clauses of width at most \(f\).
For \(f=1\), every constraint is simply a required separation
\(x\not\equiv_\theta y\).
\end{remark}

\subsection{Relational-morphism characterization}

\begin{lemma}[Free-monoid selector for a relational morphism]
\label{alg:lem:selector}
Let \(\rho:T\relto M\) be a monoid relational morphism.  For each letter
\(a\in\Sigma\), choose
\(m_a\in\rho(\eta(a)).\)
There is a unique monoid homomorphism
\(h_\rho:\Sigma^*\to M, \qquad h_\rho(a)=m_a\quad(a\in\Sigma),\)
and it satisfies
\(h_\rho(w)\in\rho(\eta(w)) \qquad(w\in\Sigma^*).\)
Consequently,
\(\rho_{h_\rho}(t)\subseteq\rho(t) \qquad(t\in T).\)
\end{lemma}

\begin{proof}
This is the standard free-monoid selection argument for relational morphisms;
we include it because the image-size comparison below uses the inclusion of
fibers explicitly.  The unique extension exists because \(\Sigma^*\) is the
free monoid on \(\Sigma\).

We prove
\(h_\rho(w)\in\rho(\eta(w))\)
by induction on \(|w|\).  For \(w=\varepsilon\),
\(h_\rho(\varepsilon)=1_M\in\rho(1_T)=\rho(\eta(\varepsilon)).\)
For \(w=va\), the induction hypothesis and the chosen letter value give
\(h_\rho(v)\in\rho(\eta(v)), \qquad h_\rho(a)=m_a\in\rho(\eta(a)).\)
The relational-morphism multiplication law therefore yields
\(h_\rho(w) =h_\rho(v)h_\rho(a) \in\rho(\eta(v)\eta(a)) =\rho(\eta(w)).\)
If \(m\in\rho_{h_\rho}(t)\), then \(m=h_\rho(w)\) for some word with
\(\eta(w)=t\); hence \(m\in\rho(t)\).  Therefore
\(\rho_{h_\rho}(t)\subseteq\rho(t).\)
\end{proof}

\begin{theorem}[Relational-morphism characterization]
\label{alg:thm:relational-characterization}
For every regular language \(L\) with syntactic pointed monoid \((T,P)\),
\[
\cmp_f(L)
=
\min\left\{
|M|:
\begin{array}{l}
M\text{ is a finite monoid and}\\
\rho:T\relto M\text{ is an \(f\)-separating relational morphism}
\end{array}
\right\}.
\]
\end{theorem}

\begin{proof}
The feasible set on the right is nonempty.  Indeed, the identity homomorphism
\(T\to T\), viewed as a functional relational morphism, is \(f\)-separating:
for an unsafe pair \((\mathbf s,\mathbf t)\), not all coordinates can be equal,
since equality in every coordinate would give identical profiles.

First note that a relational morphism has no reason to carry unused codomain
values.  If \(\rho:T\relto M\), then
\(M':=\pi_2(\operatorname{Gr}(\rho))=\bigcup_{t\in T}\rho(t)\) is a
submonoid of \(M\): it contains \(1_M\), and the relational multiplication
law closes it under products.  Replacing \(M\) by \(M'\) changes none of the
fibers and hence preserves the separating property.  Thus the minimum on the
right-hand side may always be taken over codomains consisting entirely of
values actually used by \(\rho\).

Let \(h:\Sigma^*\to M\) be a minimum witnessing compression map, with codomain
restricted to \(\operatorname{im}(h)\).
By Theorem~\ref{alg:thm:safety-interface}, the canonical relation
\(\rho_h:T\relto\operatorname{im}(h)\)
is \(f\)-separating. Therefore
\(\min_{\rho}|M| \le \cmp_f(L).\)
Conversely, let \(\rho:T\relto M\) be any finite \(f\)-separating relational morphism.
Choose generator values and form \(h_\rho:\Sigma^*\to M\) as in Lemma~\ref{alg:lem:selector}.
Since
\(\rho_{h_\rho}(t)\subseteq\rho(t)\)
for every \(t\), every disjointness
\(\rho(s_i)\cap\rho(t_i)=\varnothing\)
implies
\(\rho_{h_\rho}(s_i)\cap\rho_{h_\rho}(t_i)=\varnothing.\)
Thus \(\rho_{h_\rho}\) is also \(f\)-separating.
By Theorem~\ref{alg:thm:safety-interface},
\(h_\rho\) witnesses tuple substitutability.

Restricting its codomain to
\(\operatorname{im}(h_\rho)\) gives \(\cmp_f(L) \le |\operatorname{im}(h_\rho)| \le |M|.\) Taking the minimum over relational morphisms yields the reverse inequality.
\end{proof}

\begin{corollary}[Syntactic invariance]
\label{alg:cor:syntactic-invariance}
Let \(L\) and \(K\) be regular languages, possibly over different finite
alphabets.  If their syntactic pointed monoids are isomorphic, then for every
finite \(f\ge1\),
\(\cmp_f(L)=\cmp_f(K).\)
Thus the entire compression sequence is an invariant of the finite pointed
syntactic monoid \((T_L,P_L)\) up to isomorphism.
\end{corollary}

\begin{proof}
The right-hand side of Theorem~\ref{alg:thm:relational-characterization}
depends only on the pointed monoid: a pointed-monoid isomorphism transports
unsafe syntactic profiles and, by composition, transports
\(f\)-separating relational morphisms in both directions without changing
the codomain size.
\end{proof}

The characterization is effective on a finite pointed syntactic monoid:
the identity gives \(\cmp_f(L)\le |T_L|\), the unsafe sets through arity
\(f\) are finite and computable, and one can enumerate finite target monoids
and relational morphisms up to that size.  No efficiency claim is implied.
The selector lemma also explains why multivalued fibers create no consistency
problem: choosing values on generators extends uniquely by freeness, remains
inside the relational fibers by multiplicativity, and can only shrink fiber
intersections.  Thus every separating relational morphism contains a
functional compression witness of no larger image size.

\subsection{A core finite-side bound: commutative collapse}

\begin{theorem}[Commutative-syntactic arity collapse]
\label{group:thm:commutative-collapse}
Let \(L\) be regular and suppose its syntactic monoid \(T_L\) is commutative.
For every finite compression map \(h\) and every \(f\ge1\),
\(
L\text{ is }(1,h)\text{-substitutable}
\iff
L\text{ is }(f,h)\text{-tuple-substitutable}.
\)
Consequently \(\cmp_1(L)=\cmp_2(L)=\cdots\).
\end{theorem}

\begin{proof}
Only the forward implication needs proof. Let
\(\vec x=(x_1,\ldots,x_d)\) and \(\vec y=(y_1,\ldots,y_d)\), \(d\le f\),
have equal componentwise \(h\)-type and share an accepting tuple context
\(E=u_0\blank_1u_1\cdots\blank_du_d.\)
Put \(X=x_1\cdots x_d\), \(Y=y_1\cdots y_d\), and
\(U=u_0\cdots u_d\). Componentwise compression equality gives
\(h(X)=h(Y)\). Since the syntactic monoid is commutative, \(\eta(E[\vec x])=\eta(UX), \qquad \eta(E[\vec y])=\eta(UY).\) Thus \(X,Y\) share the accepting unary context \(U\blank_1\), and unary
\(h\)-substitutability gives \(\D_L(X)=\D_L(Y)\).

For any other tuple context
\(F=v_0\blank_1v_1\cdots\blank_dv_d,\)
put \(V=v_0\cdots v_d\). Again by commutativity, \(F[\vec x]\in L \iff VX\in L \iff VY\in L \iff F[\vec y]\in L.\) Hence the tuple distributions are equal.
\end{proof}

\section{Arity Amplification inside \texorpdfstring{\(\mathbf V_{ab}\)}{Vab}}
\label{sec:vab-arity}

The center-index theorem,
Theorem~\ref{group:thm:center-index}, shows that the first arity jump can
already be large.  We now show that every successive arity boundary can be
arbitrarily large while the syntactic monoid remains in one fixed classical
pseudovariety.

Let
\(
M_{ab}:=\operatorname{Synt}(\{ab\})
=\{1,a,b,ab,0\},
\qquad
\mathbf V_{ab}:=\langle M_{ab}\rangle.
\)
The monoid \(M_{ab}\) is one of the canonical aperiodic generators occurring in
the Margolis--Pin classification of minimal noncommutative
varieties~\cite{MargolisPin1984}.

Call a word \(w\in\Delta^*\) \emph{letter-distinct} if every letter of
\(\Delta\) occurs in \(w\) at most once. A language is letter-distinct if
each of its words is letter-distinct.  Margolis and Pin call such words
\emph{multilinear}; their Proposition~1.1 characterizes the presence of
\(M_{ab}\) in a variety by, equivalently, the inclusion of all multilinear
singleton languages in the corresponding variety of languages
\cite[Proposition~1.1]{MargolisPin1984}.  The lemma below is therefore also a
finite-union consequence of that result; we include an explicit recognition
argument because the formulation in terms of \(T_L\in\mathbf V_{ab}\) is used
repeatedly below.

\begin{lemma}[Letter-distinct finite languages lie in \(\mathbf V_{ab}\)]
\label{vab:lem:letterdistinct}
Let \(\Delta\) be finite. If \(L\subseteq\Delta^*\) is finite and every word
of \(L\) has pairwise distinct letters, then \(T_L\in\mathbf V_{ab}.\) \end{lemma}

\begin{proof}
Let \(w=c_1\cdots c_m\) have pairwise distinct letters. For each \(i\), map
\(c_i\mapsto ab\) and every other letter to \(1\) in \(M_{ab}\); requiring image
\(ab\) forces \(c_i\) to occur exactly once. For each \(i<m\), map
\(c_i\mapsto a\), \(c_{i+1}\mapsto b\), and every other letter to \(1\);
under the exact-once conditions, requiring image \(ab\) forces \(c_i\) to
precede \(c_{i+1}\). A further factor sends letters not occurring in \(w\)
to \(0\). A finite direct product therefore recognizes exactly \(\{w\}\),
so its syntactic monoid divides a finite power of \(M_{ab}\). For the finite
union \(L\), take the direct product of these singleton recognizers and
choose as accepting subset the product states for which at least one
singleton coordinate is accepting.  This is a well-defined subset of the
product monoid because each coordinate state already determines membership in
the corresponding singleton accepting set.  The product morphism therefore
recognizes exactly \(L\), so \(T_L\) divides a finite power of \(M_{ab}\).
\end{proof}

Fix \(d\ge1\), put \(n=d+1\), and for \(1\le j\le R\) introduce pairwise
distinct letters
\(a_{j,1},\ldots,a_{j,n}, \qquad c_{j,1},\ldots,c_{j,n-1}.\)
Define \(B_j=a_{j,1}\cdots a_{j,n},\) \(
D_j=a_{j,1}c_{j,1}a_{j,2}c_{j,2}\cdots
c_{j,n-1}a_{j,n},
\)
and
\(\widehat L_{d,R}:=\{B_j,D_j:1\le j\le R\}.\)
Every accepted word is letter-distinct, hence
\(T_{\widehat L_{d,R}}\in\mathbf V_{ab}\).

\begin{lemma}[Bounded compression through arity \(d\)]
\label{vab:lem:low}
For \(n=d+1\), there is a compression map \(h_d\), independent of \(R\), such
that \(\widehat L_{d,R}\) is \((d,h_d)\)-tuple-substitutable and
\(
  |\operatorname{im}(h_d)|\le
  C_d:=\frac{5d^2+9d+8}{2}.
\)
In particular, \(\cmp_d(\widehat L_{d,R})\le C_d\).
\end{lemma}

\begin{proof}
Let
\(\Gamma_d=\{A_1,\ldots,A_n,C_1,\ldots,C_{n-1}\}\), and let
\(\rho_d:\Sigma_{d,R}^*\to\Gamma_d^*\) send
\(a_{j,i}\mapsto A_i\) and \(c_{j,i}\mapsto C_i\).  The two accepted role
templates are
\(\beta=A_1\cdots A_n, \qquad \delta=A_1C_1A_2C_2\cdots C_{n-1}A_n.\)
Put \(F_d=\operatorname{Fact}(\{\beta,\delta\})\), including the empty
factor.  Since \(\Gamma_d^*\setminus F_d\) is a two-sided ideal, the Rees
factor
\(Q_d=F_d\cup\{0\}\)
is a finite monoid: factors multiply by concatenation when the concatenation
is again a factor and otherwise to \(0\).  Let \(q_d:\Gamma_d^*\to Q_d\)
be the quotient morphism and set \(h_d=q_d\circ\rho_d\).  The longest
template has length \(2n-1\); hence, using the crude sum of the numbers of
nonempty factors of \(\beta\) and \(\delta\),
\(|Q_d| \le 2+\frac{n(n+1)}2+n(2n-1) =\frac{5d^2+9d+8}{2}=C_d.\)
The quotient map is injective on \(F_d\).

Let \(1\le r\le d=n-1\), and let
\(\mathbf x=(x_1,\ldots,x_r)\) and
\(\mathbf y=(y_1,\ldots,y_r)\) have equal componentwise \(h_d\)-type and
share an accepting tuple context \(E\).  Every component occurring in an
accepted filling is a factor of an accepted word, so its role word lies in
\(F_d\).  Equality of \(h_d\)-types and injectivity on \(F_d\) therefore
give
\(\rho_d(x_i)=\rho_d(y_i)\qquad(1\le i\le r).\)
Applying \(\rho_d\) to the two accepted fillings shows that they have the
same role template, either \(\beta\) or \(\delta\).

If the two fillings belong to the same copy \(j\), a factor of \(B_j\) or
\(D_j\) is uniquely determined by its role word, so \(x_i=y_i\) for every
\(i\).  Their complete tuple distributions are therefore equal.

Assume now that the fillings belong to distinct copies \(j\ne k\).
If every tuple component were empty, then \(\mathbf x=\mathbf y\) and the two
completed fillings would coincide, so they could not belong to distinct
copies.  Hence at least one tuple component is nonempty.  Since the copy
alphabets are disjoint, any fixed terminal in the shared accepting context
would force one completed word to mix two copy alphabets.  Hence
\(E\) has no fixed terminal material.  Its nonempty tuple components thus
partition the whole accepted word, in the order in which the holes occur,
and the two fillings have the same role template.

Suppose first that this template is \(\delta\).  Then the tuple
components have total length \(|\delta|=2n-1\), the maximum length of an
accepted word.  Let \(E'\) be an arbitrary tuple context.  If
\(E'[\mathbf x]\in\widehat L_{d,R}\), then \(E'\) cannot contain any
fixed terminal material, because the inserted tuple already has total length
\(2n-1\).  Hence the filled word is simply the concatenation of the tuple
components in their fixed hole order.  Since the \(x_i\) use only the copy-%
\(j\) alphabet, this word must be \(D_j\).  The equal role blocks of
\(\mathbf y\) then concatenate to \(D_k\), so
\(E'[\mathbf y]\in\widehat L_{d,R}\).  The same argument with
\(j\) and \(k\) interchanged proves the converse.  Therefore
\(E'[\mathbf x]\in\widehat L_{d,R} \quad\Longleftrightarrow\quad E'[\mathbf y]\in\widehat L_{d,R}\)
for every \(E'\), and the two tuple distributions are equal.

Finally suppose that the common template is \(\beta\).  The nonempty tuple
components partition \(A_1\cdots A_n\) into at most \(r\le n-1\)
contiguous role blocks.  Hence at least one block contains an adjacent pair
\(A_iA_{i+1}\).  Again let \(E'\) be arbitrary and assume first that
\(E'[\mathbf x]\in\widehat L_{d,R}\).  Because the inserted letters all
belong to copy \(j\), an accepted filling cannot mix them with fixed letters
from another copy.  If the target is \(B_j\), its length is already the total
inserted length \(n\), so \(E'\) has no fixed terminal material; replacing
copy \(j\) by copy \(k\) then gives \(B_k\), and
\(E'[\mathbf y]\) is accepted.

The only other possible target over the copy-%
\(j\) alphabet would be \(D_j\), whose role template is \(\delta\).
But \(D_j\) places the marker \(C_i\) strictly between
\(A_i\) and \(A_{i+1}\), while those two roles lie inside one tuple
component.  Fixed context material cannot be inserted into the interior of a
component, so \(D_j\) cannot be formed.  Thus every context accepting
\(\mathbf x\) also accepts \(\mathbf y\).  Interchanging \(j\) and
\(k\) gives the reverse implication.  Hence again
\(E'[\mathbf x]\in\widehat L_{d,R} \quad\Longleftrightarrow\quad E'[\mathbf y]\in\widehat L_{d,R}\)
for every tuple context \(E'\).

All cases give equality of the complete tuple distributions, so \(h_d\) is a
witness through arity \(d\).
\end{proof}

\begin{lemma}[Unbounded next arity]
\label{vab:lem:high}
For \(n=d+1\), \(\cmp_n(\widehat L_{d,R}) \ge \left\lceil R^{1/n}\right\rceil.\) \end{lemma}

\begin{proof}
Put \(\mathbf a_j=(a_{j,1},\ldots,a_{j,n}).\) The adjacent-hole context \(\blank_1\cdots\blank_n\) accepts every
\(\mathbf a_j\), producing \(B_j\). The context
\(\blank_1c_{j,1}\blank_2c_{j,2}\cdots c_{j,n-1}\blank_n\)
accepts \(\mathbf a_j\), producing \(D_j\), and rejects every
\(\mathbf a_k\) with \(k\ne j\). Thus the \(R\) tuples are pairwise unsafe.
If a witness has \(m\) image values, their componentwise types lie in a set of
size at most \(m^n\) and must be distinct. Hence \(R\le m^n\).
\end{proof}

\begin{theorem}[Every-boundary amplification inside \(\mathbf V_{ab}\)]
\label{vab:thm:amplification}
For every fixed \(d\ge1\),
\(\sup_{\,T_L\in\mathbf V_{ab}} \frac{\cmp_{d+1}(L)}{\cmp_d(L)} =\infty.\)
The witnesses are finite letter-distinct languages whose words have length at
most \(2d+1\).  Moreover, there are an integer \(m_d\ge1\) and a sequence
of such languages \(L_j\), with \(T_{L_j}\in\mathbf V_{ab}\), for which
\(\cmp_d(L_j)=m_d \quad\text{for all }j, \qquad \cmp_{d+1}(L_j)\longrightarrow\infty.\)
\end{theorem}

\begin{proof}
For \(\widehat L_{d,R}\), Lemmas~\ref{vab:lem:low} and
\ref{vab:lem:high} give
\(\frac{\cmp_{d+1}(\widehat L_{d,R})} {\cmp_d(\widehat L_{d,R})} \ge \frac{\lceil R^{1/(d+1)}\rceil}{C_d},\)
which tends to infinity with \(R\).  The same lemmas give
\(1\le\cmp_d(\widehat L_{d,R})\le C_d\) and
\(\cmp_{d+1}(\widehat L_{d,R})\ge\lceil R^{1/(d+1)}\rceil\).
Hence \(\cmp_d(\widehat L_{d,R})\) takes only finitely many integer values;
one value occurs for infinitely many \(R\), and along that subsequence the
lower-arity value is constant while the next-arity value diverges.
\end{proof}

\subsection{Finite-prefix universality inside \texorpdfstring{\(\mathbf V_{ab}\)}{Vab}}

\begin{lemma}[Disjoint-union bounds]
\label{vab:lem:union}
Let \(L_i\subseteq\Sigma_i^*\) be languages over pairwise disjoint finite
alphabets, assume \(\varepsilon\notin L_i\) for every \(i\), and put
\(L=\bigcup_{i=1}^sL_i\). Then
\(\max_i\cmp_f(L_i)\le\cmp_f(L)\). If
\(h_i:\Sigma_i^*\to M_i\) witnesses \(L_i\), then \(L\) has a witness with at
most \(2^s\prod_{i=1}^s|M_i|\) image elements.
\end{lemma}

\begin{proof}
For the lower bound, let \(h\) witness \(L\) and fix \(i\). Since the
alphabets are disjoint and no \(L_j\) contains the empty word,
\(L_i=L\cap\Sigma_i^*\). Every unsafe tuple pair for \(L_i\), witnessed and
distinguished by contexts over \(\Sigma_i\), is therefore still unsafe for
\(L\). Hence the
restriction of \(h\) to \(\Sigma_i^*\) witnesses \(L_i\), giving
\(
\cmp_f(L_i)\le |\operatorname{im}(h)|.
\)

For the upper bound, let
\(\pi_i:\Sigma^*\to\Sigma_i^*\) erase all letters outside \(\Sigma_i\), and
let
\(
\operatorname{supp}:\Sigma^*\longrightarrow(2^{[s]},\cup)
\)
record the set of component alphabets occurring in a word. Define
\(
H=(\operatorname{supp},h_1\circ\pi_1,\ldots,h_s\circ\pi_s).
\)
Its image has size at most \(2^s\prod_i|M_i|\).

Let coordinatewise \(H\)-equal tuples \(\mathbf x,\mathbf y\) share an
accepting tuple context \(E\). If every tuple component is empty then
\(\mathbf x=\mathbf y\), so assume that some component is nonempty. Since
\(E[\mathbf x]\in L\), it lies in a unique \(L_j\) because the alphabets are
pairwise disjoint. Consequently every fixed terminal of \(E\) and every
nonempty \(x_k\) uses only letters of \(\Sigma_j\). Equality of the support
coordinates implies that each corresponding \(y_k\) has exactly the same
alphabet support; hence \(E[\mathbf y]\) also lies in \(L_j\).
Componentwise equality of \(H\) now gives
\(
h_j(\pi_j(x_k))=h_j(\pi_j(y_k))
\qquad\text{for every coordinate }k,
\)
and here \(\pi_j(x_k)=x_k\) and \(\pi_j(y_k)=y_k\). Since \(h_j\) witnesses
\(L_j\), the two tuples have the same distribution with respect to all named
contexts over \(\Sigma_j\).

It remains to compare contexts over the full union alphabet. Any context
\(F\) accepting one of the two tuples must again complete it to some \(L_k\).
The presence of a nonempty tuple component, together with equality of support,
forces \(k=j\); thus all fixed terminals of such an accepting \(F\) lie in
\(\Sigma_j\). Therefore \(F\) is effectively a \(\Sigma_j\)-context and
acceptance is the same for the two tuples by the preceding paragraph. A
context mixing component alphabets accepts neither tuple. Hence the complete
tuple distributions in \(L\) are equal, and \(H\) witnesses \(L\).
\end{proof}

\begin{lemma}[Triangular gluing principle]
\label{vab:lem:triangular-gluing}
Fix \(m\ge2\).  For each \(1\le g<m\), let
\(\{K_{g,r}:r\ge1\}\) be a family of \(\varepsilon\)-free regular languages
over finite alphabets.  Assume that there is an integer \(C_g\ge1\) and a
function \(B_g:\mathbb N\to\mathbb R_{>0}\) with
\(B_g(r)\to\infty\) such that, for every \(r\),
\[
 \cmp_g(K_{g,r})\le C_g,
 \qquad
 \cmp_{g+1}(K_{g,r})\ge B_g(r).
 \tag{TG}
\]
Then for arbitrary \(\lambda_1,\ldots,\lambda_{m-1}>0\), one can choose
parameters \(r_1,\ldots,r_{m-1}\), rename the component alphabets to be
pairwise disjoint, and form
\(K:=\bigcup_{g=1}^{m-1}K_{g,r_g}\)
so that
\(\cmp_{f+1}(K)>\lambda_f\cmp_f(K) \quad\text{and}\quad \cmp_{f+1}(K)>\cmp_f(K) \qquad(1\le f<m).\)
In particular, any collection of boundary amplifiers satisfying \textup{(TG)}
can be glued into one language whose successive compression ratios
simultaneously exceed arbitrarily prescribed finite multiplicative gap
thresholds.
\end{lemma}

\begin{proof}
For a fixed component \(K_{g,r}\), exact syntactic recognition supplies an
all-arity witness of size
\(D_g(r):=|T_{K_{g,r}}|.\)
After renaming alphabets, which does not change any compression number, all
selected components may be assumed pairwise alphabet-disjoint.

Choose \(r_1,\ldots,r_{m-1}\) recursively.  When
\(r_1,\ldots,r_{f-1}\) have been fixed, set
\(U_f := 2^{m-1} \left(\prod_{g<f}D_g(r_g)\right) \left(\prod_{g\ge f}C_g\right).\)
Crucially, \(U_f\) depends on already chosen parameters only: components with
index \(g\ge f\) enter through the uniform lower-arity bounds \(C_g\), not
through their future parameters.  Since \(B_f(r)\to\infty\), choose \(r_f\)
so that
\(B_f(r_f)>\max\{1,\lambda_f\}U_f.\)
This is the triangular feature that makes the recursion non-circular.

Let \(K=\bigcup_{g=1}^{m-1}K_{g,r_g}\).  At arity \(f\), use the syntactic
all-arity witness of size \(D_g(r_g)\) for each already chosen component
\(g<f\).  For \(g\ge f\), condition \textup{(TG)} gives a witness through
arity \(g\), hence through arity \(f\), with at most \(C_g\) values.
The upper half of Lemma~\ref{vab:lem:union} therefore gives
\(\cmp_f(K)\le U_f.\)
On the other hand, the lower half of the same lemma and the \(f\)-th amplifier
give
\[
\begin{aligned}
 \cmp_{f+1}(K)
 &\ge \cmp_{f+1}(K_{f,r_f})
 \ge B_f(r_f)\\
 &>\max\{1,\lambda_f\}U_f
 \ge\max\{1,\lambda_f\}\cmp_f(K).
\end{aligned}
\]
This proves both claimed inequalities simultaneously for every
\(1\le f<m\).
\end{proof}

\begin{theorem}[Finite-prefix universality inside \(\mathbf V_{ab}\)]
\label{vab:thm:finite-prefix}
Let \(m\ge2\) and \(\lambda_1,\ldots,\lambda_{m-1}>0\).  There exists a
finite letter-distinct language \(L\) with \(T_L\in\mathbf V_{ab}\), all of whose
words have length at most \(2m-1\), such that
\(\cmp_{f+1}(L)>\lambda_f\cmp_f(L) \qquad(1\le f<m).\)
In particular,
\(\cmp_1(L)<\cmp_2(L)<\cdots<\cmp_m(L).\)
\end{theorem}

\begin{proof}
For each \(1\le g<m\), apply Lemmas~\ref{vab:lem:low} and
\ref{vab:lem:high} to the amplifier family
\(K_{g,R}:=\widehat L_{g,R}.\)
These languages are finite and \(\varepsilon\)-free, and they satisfy the
hypotheses of Lemma~\ref{vab:lem:triangular-gluing} with
\(C_g=\frac{5g^2+9g+8}{2}, \qquad B_g(R)=\left\lceil R^{1/(g+1)}\right\rceil.\)
The triangular gluing principle therefore chooses
\(R_1,\ldots,R_{m-1}\) such that, after pairwise disjoint renaming of the
alphabets,
\(L=\bigcup_{g=1}^{m-1}\widehat L_{g,R_g}\)
satisfies all the required inequalities.

Every accepted word lies in one component and is letter-distinct; its length
is at most
\(
 \max_{g<m}(2g+1)=2m-1.
\)
Hence the union is again a finite letter-distinct language, and
Lemma~\ref{vab:lem:letterdistinct} gives
\(T_L\in\mathbf V_{ab}\).
\end{proof}

\subsection{SCL compression height}

\begin{definition}[SCL compression height]
\label{ah:def:height}
Let \(\mathbf V\) be a pseudovariety of finite monoids. Since
\(\Obj_1(L)\cong T_L\), define
\[
\operatorname{ch}(\mathbf V)
:=
\min\left\{
r\in\mathbb{N}_{\ge1}:
\begin{array}{l}
\text{for every regular language \(L\) with \(T_L\in\mathbf V\),}\\
\cmp_f(L)=\cmp_r(L)
\text{ for every }f\in\mathbb{N}_{\ge r}
\end{array}
\right\},
\]
with the convention that the minimum of the empty set is \(\infty\).
Here and throughout this definition, the quantification over regular
languages ranges over languages over arbitrary finite alphabets; no common
alphabet is fixed across the class.
\end{definition}

For each fixed regular language \(L\), the sequence
\(\cmp_1(L)\le\cmp_2(L)\le\cdots\le |T_L|\) is nondecreasing and bounded,
so it always stabilizes at some finite arity.  Thus
\(\operatorname{ch}(\mathbf V)=\infty\) means that no single stabilization
arity works uniformly for all languages with syntactic monoid in \(\mathbf
V\); it does not mean that one language has a divergent compression sequence.

\begin{lemma}[Monotonicity]
\label{ah:lem:monotone}
If \(\mathbf V\subseteq\mathbf W\), then
\(
\operatorname{ch}(\mathbf V)
\le
\operatorname{ch}(\mathbf W).
\)
\end{lemma}

\begin{proof}
The class of languages quantified over for \(\mathbf V\) is contained in the
class quantified over for \(\mathbf W\).
\end{proof}

\begin{corollary}[Two core height bounds]
\label{height:cor:core-bounds}
The pseudovariety \(\mathbf{Com}\) of finite commutative monoids has
\(\operatorname{ch}(\mathbf{Com})=1,\)
whereas
\(\operatorname{ch}(\mathbf V_{ab})=\infty.\)
\end{corollary}

\begin{proof}
The first statement is Theorem~\ref{group:thm:commutative-collapse}.  The
second follows from Theorem~\ref{vab:thm:amplification}: every arity boundary
admits an unbounded gap inside \(\mathbf V_{ab}\), so no finite stabilization
arity works uniformly on that pseudovariety.
\end{proof}

\section{Completely Regular Monoids Saturate at Arity Two}
\label{sec:cr-saturation}

\paragraph{Proof strategy.}
By Lemma~\ref{alg:lem:relational-overlap}, the arity problem for relational
morphisms is the corresponding problem for compatible tolerances.  Although
such a tolerance need not be transitive, binary safety forces strong rigidity
inside every accepting completely simple component.  Related elements differ
there by a unique central translation; these local discrepancies assemble
into a homomorphism
\(z_\beta:\Gamma_\beta\to N_\beta\le Z(G_\beta)\).
A shared accepting context makes the product of the coordinate defects
trivial, and contextual factorization preserves that identity in every other
tuple context.  This upgrades binary safety to all finite arities.

Let $(S,P)$ be a finite syntactic pointed monoid.  We repeatedly use the
syntactic separation property
\[
 r\ne s
 \quad\Longrightarrow\quad
 \text{there exist }a,b\in S\text{ such that exactly one of }arb,asb\text{ lies in }P.
\tag{Syn}
\]
This is the defining separation of distinct syntactic classes transported
through the surjective syntactic morphism.

We use the typed safety dictionary of
Table~\ref{tab:typed-safety-dictionary}.  In particular,
Lemma~\ref{alg:lem:relational-overlap} turns any relational witness \(\rho\)
into its overlap tolerance \(\tau_\rho\), with \(f\)-separation of \(\rho\)
exactly equivalent to \(f\)-safety of \(\tau_\rho\).  Hence it is enough---and
slightly stronger---to prove that every \(2\)-safe compatible tolerance on a
completely regular syntactic pointed monoid is safe at all finite arities.
No transitivity is assumed or used.  In the classical terminology such a
relation is a tolerance
\cite{ZelinkaTolerance1975,Pondelicek1978,KumaresanTolerance1984}; in the
direct-square language it is a symmetric diagonal subsemigroup of
\(S\times S\)~\cite{BarberRuskucDirectSquare,BarberRuskucDiagonal2026}.

We work directly with the finite syntactic profiles of
Lemma~\ref{alg:lem:profile}: a ``tuple context'' may therefore have fixed
factors in $S$ rather than concrete words.  Surjectivity of the syntactic
morphism gives a concrete word representative for every such profile.

Throughout this section assume that $S$ is completely regular.  We use
Clifford's decomposition of a completely regular semigroup into a semilattice
of completely simple components~\cite[Theorem~4.1.3]{Howie1995}, the
Rees--Suschkewitsch representation of each component
\cite[Theorem~3.3.1]{Howie1995}, and normalized Rees coordinates
\cite[Theorem~3.4.2]{Howie1995}.  The proof below uses only the normalized
Rees multiplication inside each accepting component; no translational-hull
normal form is needed.

Thus
\(S=\bigsqcup_{\alpha\in Y}S_\alpha,\) where $Y=S/\mathcal D$ is a semilattice,
every $S_\alpha$ is completely simple, and
\[
S_\alpha S_\beta\subseteq S_{\alpha\beta}.
\tag{CR1}
\]
We use the natural order
\(\beta\le\alpha \Longleftrightarrow \alpha\beta=\beta.\)
For $x\in S_\alpha$ write $\operatorname{supp}(x)=\alpha$.  Since $S$ is a
monoid, $1_Y:=\operatorname{supp}(1_S)$ is the top element of $Y$.

\begin{lemma}[Lower-component action and sandwich ideals]
\label{cr:lem:sandwich-ideal}
Let \(c\in S_\alpha\) and let \(\beta\le\alpha\). Then
\(
cS_\beta\subseteq S_\beta,
\qquad
S_\beta c\subseteq S_\beta.
\)
Consequently \(S_\beta c S_\beta\) is a nonempty two-sided ideal of the
completely simple semigroup \(S_\beta\), and hence
\(S_\beta c S_\beta=S_\beta.\)
\end{lemma}

\begin{proof}
By \textup{(CR1)} and \(\beta\le\alpha\),
\(
S_\alpha S_\beta\subseteq S_{\alpha\beta}=S_\beta,
\qquad
S_\beta S_\alpha\subseteq S_{\beta\alpha}=S_\beta.
\)
Hence \(S_\beta cS_\beta\) is nonempty.  If $z=xcy$ belongs to it and
$r\in S_\beta$, then $rz=(rx)cy$ and $zr=xc(yr)$ also belong to it.  Thus it
is a two-sided ideal of $S_\beta$; simplicity of a completely simple
semigroup gives the claim.
\end{proof}

The proof has two interfaces.  The semantic interface consists of the next
two lifting principles; the structural interface is the central-defect
homomorphism proved below.

\begin{lemma}[Syntactic-value lifting from safe overlap]
\label{cr:lem:syntactic-lifting}
Let $1\le d\le f$ and let $\tau$ be an $f$-safe compatible tolerance.  Suppose
$s_i\tau t_i$ for $1\le i\le d$ and the two $d$-tuples share an accepting
tuple context.  Then for every $d$-ary tuple context $F$,
\(F[\mathbf s]=F[\mathbf t]\quad\text{in }S.\)
\end{lemma}

\begin{proof}
Because the two tuples are coordinatewise $\tau$-related and share an
accepting profile, $f$-safety rules out their forming an element of
$\mathcal U_d(S,P)$.  Hence their finite profiles are equal, and therefore
$F[\mathbf s]\in P$ iff $F[\mathbf t]\in P$ for every $F$.  If for some $F$
the two syntactic values were distinct, \textup{(Syn)} would give $a,b\in S$
such that exactly one of $aF[\mathbf s]b$ and $aF[\mathbf t]b$ lies in $P$.
The wrapped context $aFb$ would then distinguish their finite profiles, a
contradiction.
\end{proof}

\begin{lemma}[Equalizing-context principle]
\label{cr:lem:equalizing-context}
Let $1\le d\le f$ and let $\tau$ be an $f$-safe compatible tolerance.
Suppose $s_i\tau t_i$ for $1\le i\le d$.  If some $d$-ary tuple context
$E_0$ satisfies
\[
E_0[\mathbf s]=E_0[\mathbf t]=c
\]
and there is an accepting component $S_\beta$ with
$\beta\le\operatorname{supp}(c)$, then
\[
F[\mathbf s]=F[\mathbf t]
\]
for every $d$-ary tuple context $F$.
\end{lemma}

\begin{proof}
Choose $p\in P\cap S_\beta$.  By
Lemma~\ref{cr:lem:sandwich-ideal}, $S_\beta cS_\beta=S_\beta$, so there are
$a,b\in S_\beta$ with $acb=p$.  Hence $aE_0b$ is a shared accepting context
for $\mathbf s$ and $\mathbf t$.  Lemma~\ref{cr:lem:syntactic-lifting}
therefore gives the asserted equality for every $F$.
\end{proof}

\begin{remark}[Proof dependencies for the arity-two theorem]
\label{cr:rem:dependencies}
The proof has a short structural spine:
\[
\begin{gathered}
\text{syntactic lifting + local rigidity}
\Longrightarrow
\text{free central translations}
\Longrightarrow
z_\beta:\Gamma_\beta\to N_\beta
\\[-1pt]
\Longrightarrow
\text{product compression}
\Longrightarrow
\text{all-arity safety}.
\end{gathered}
\]
More precisely, Lemmas~\ref{cr:lem:syntactic-lifting} and
\ref{cr:lem:equalizing-context} form the semantic interface.  Together with
Lemmas~\ref{cr:lem:local-kernel}, \ref{cr:lem:local-central}, and
\ref{cr:lem:index-rigidity}, they give the local stage.  Its uniform form is
the central-translation action of Lemma~\ref{cr:lem:central-translations}
and the homomorphism of Proposition~\ref{cr:lem:uniform-defect}.  Contextual
factorization and Lemma~\ref{cr:lem:product-compression} then yield
Theorem~\ref{cr:thm:tolerance-saturation}.  Support bookkeeping enters only
at the product-compression stage.  No classification of tolerances or
diagonal subsemigroups, no translational-hull normal form, and no safety
assumption above arity two is used.
\end{remark}

\subsection{Local rigidity inside an accepting completely simple component}

For each \(\beta\in Y\) with \(P\cap S_\beta\ne\varnothing\), fix once and
for all a normalized Rees representation
\(S_\beta=\mathcal M[G_\beta;I_\beta,\Lambda_\beta;Q^\beta].\)
In the local arguments below, \(\beta\) denotes an arbitrary accepting
component, and all Rees-coordinate calculations are carried out in its fixed
normalized representation.
Throughout the section we use the multiplication convention
\[
(i,g,\lambda)(j,h,\mu)
=
(i,\,gq^\beta_{\lambda j}h,\,\mu).
\tag{CR2}
\]
The distinguished row and column are indexed by $1$, and normalization means
\(
q^\beta_{1i}=q^\beta_{\lambda 1}=1_{G_\beta}
\qquad(i\in I_\beta,\ \lambda\in\Lambda_\beta).
\)
All later Rees-coordinate calculations use \textup{(CR2)}.
The restriction
\(\tau_\beta:=\tau\cap(S_\beta\times S_\beta)\)
is a compatible tolerance on $S_\beta$, equivalently a symmetric diagonal
subsemigroup of $S_\beta\times S_\beta$: it is reflexive and symmetric by
construction, and compatibility was checked above.  As for $\tau$ itself,
$\tau_\beta$ need not be transitive.

Write
\(H_{11}=\{(1,g,1):g\in G_\beta\}, \qquad e=(1,1,1), \qquad \bar g=(1,g,1).\)
For \(x=(i,g,\lambda)\in S_\beta\), normalization in \textup{(CR2)} gives
\(exe=\bar g\), \(xe=(i,g,1)\), and \(ex=(1,g,\lambda)\), the identities
used below.
Because the Rees representation is normalized, $H_{11}$ is a subgroup
isomorphic to $G_\beta$ with identity $e$.

\begin{lemma}[Local kernel]
\label{cr:lem:local-kernel}
Define
\(
N_\beta:=\{n\in G_\beta:e\tau\bar n\}.
\)
Then $N_\beta\trianglelefteq G_\beta$.
\end{lemma}

\begin{proof}
Let
\(R:=\{(x,y)\in H_{11}\times H_{11}:x\tau y\}.\)
Compatibility makes $R$ a subsemigroup of the finite group
$H_{11}\times H_{11}$, and reflexivity puts the diagonal (hence $(e,e)$) in
$R$. Every subsemigroup of a finite group that contains the identity is a
subgroup: for $r\in R$, finiteness gives $r^m=(e,e)$ for some $m>0$, so
$r^{-1}=r^{m-1}\in R$. Thus $R$ is a subgroup. Its fiber over $e$ in the
first coordinate is exactly $\{(e,\bar n):n\in N_\beta\}$, whence
$N_\beta$ contains the identity and is closed under multiplication and
inverses. Therefore $N_\beta\le G_\beta$.

For normality, let $n\in N_\beta$ and $g\in G_\beta$. The three related pairs
\((\bar g,\bar g),\qquad (e,\bar n),\qquad (\overline{g^{-1}},\overline{g^{-1}})\)
may be multiplied coordinatewise by compatibility. Their product is
\((e,\overline{gng^{-1}}),\)
so $gng^{-1}\in N_\beta$.
\end{proof}

\begin{lemma}[Local centrality]
\label{cr:lem:local-central}
If $\tau$ is $2$-safe, then
\(N_\beta\le Z(G_\beta).\)
\end{lemma}

\begin{proof}
Let $n\in N_\beta$.  By Lemma~\ref{cr:lem:local-kernel}, also
$n^{-1}\in N_\beta$, so the tuples
\((e,e)\) and \((\bar n,\overline{n^{-1}})\) are coordinatewise
$\tau$-related.  The empty-fixed binary context
$E_0=\blank_1\blank_2$ gives the same value $e$ on both tuples.  Since
$S_\beta$ is accepting, Lemma~\ref{cr:lem:equalizing-context} applies.
For arbitrary $g\in G_\beta$, evaluate it at
\(F=\blank_1\bar g\blank_2\overline{g^{-1}}.\)
Then
\[
e
=
\bar n\bar g\overline{n^{-1}}\overline{g^{-1}}
=
\overline{ngn^{-1}g^{-1}},
\]
so $ngn^{-1}g^{-1}=1$.  Thus $n\in Z(G_\beta)$.
\end{proof}

\begin{lemma}[Local row and column rigidity]
\label{cr:lem:index-rigidity}
Suppose $\tau$ is $2$-safe and
\((i,g,\lambda)\tau_\beta(j,h,\mu)\).  Then $i=j$ and $\lambda=\mu$, with
\(gh^{-1}\in N_\beta\).
\end{lemma}

\begin{proof}
Put $x=(i,g,\lambda)$ and $y=(j,h,\mu)$.  Compatibility with the
distinguished idempotent $e$ gives
\begin{equation}
\bar g=exe\;\tau\;eye=\bar h.
\label{eq:cr-compressed-related}
\end{equation}
By symmetry and compatibility,
\(e\tau\overline{h^{-1}g}\); put $d_R:=h^{-1}g\in N_\beta$.
Then the tuples $(x,e)$ and $(y,\overline{d_R})$ are coordinatewise
$\tau$-related, while the context
\(E_R=e\blank_1\blank_2\) gives the common value $\bar g$.
Lemma~\ref{cr:lem:equalizing-context}, followed by the empty-fixed context
$\blank_1\blank_2$, yields
\[
xe=y\overline{d_R}.
\]
In normalized Rees coordinates this is
$(i,g,1)=(j,g,1)$, hence $i=j$.

Dually, compatibility in \eqref{eq:cr-compressed-related} gives
$e\tau\overline{d_L}$ for $d_L:=gh^{-1}\in N_\beta$.
The tuples $(e,x)$ and $(\overline{d_L},y)$ are coordinatewise
$\tau$-related, and
\(E_L=\blank_1\blank_2e\) gives the common value $\bar g$.
Applying Lemma~\ref{cr:lem:equalizing-context} and then
$\blank_1\blank_2$ gives
\[
ex=\overline{d_L}y.
\]
Thus $(1,g,\lambda)=(1,g,\mu)$, so $\lambda=\mu$; the same calculation has
already shown $gh^{-1}=d_L\in N_\beta$.
\end{proof}

Thus every accepting completely simple component has the strong local form
\[
x\tau_\beta y
\quad\Longrightarrow\quad
x,y\text{ lie in the same }\mathcal H\text{-class and differ by }
N_\beta\le Z(G_\beta).
\tag{LR}
\]

\subsection{Uniform central defects}

Fix the accepting component $S_\beta$ from the preceding subsection. Put
\(\Omega_\beta:=\{a\in S:\beta\le\operatorname{supp}(a)\}.\)
Since support multiplication is the semilattice product, \(\Omega_\beta\) is a
submonoid of \(S\).  By Lemma~\ref{cr:lem:sandwich-ideal}, every
\(a\in\Omega_\beta\) preserves \(S_\beta\) under multiplication from either
side.

\begin{lemma}[Central translations]
\label{cr:lem:central-translations}
For \(n\in N_\beta\), define
\[
\nu_\beta(n)(i,g,\lambda):=(i,ng,\lambda)
\qquad((i,g,\lambda)\in S_\beta).
\]
Then \(\nu_\beta(mn)=\nu_\beta(m)\circ\nu_\beta(n)\) and
\(\nu_\beta(1)=\operatorname{id}_{S_\beta}\).  Moreover, for
\(x,y\in S_\beta\),
\[
\nu_\beta(n)(x)y=x\nu_\beta(n)(y)=\nu_\beta(n)(xy),
\tag{CT}
\]
the action is free, and for every \(a\in\Omega_\beta\),
\[
a\nu_\beta(n)(x)=\nu_\beta(n)(ax),
\qquad
\nu_\beta(n)(x)a=\nu_\beta(n)(xa).
\tag{CE}
\]
In particular, \(\nu_\beta(m)(x)=\nu_\beta(n)(x)\) for some \(x\in S_\beta\)
implies \(m=n\).
\end{lemma}

\begin{proof}
The composition law is immediate from the definition.  Write
\(x=(i,g,\lambda)\) and \(y=(j,h,\mu)\).  Since
\(N_\beta\le Z(G_\beta)\) by Lemma~\ref{cr:lem:local-central},
\textup{(CR2)} gives
\[
\nu_\beta(n)(x)y
=(i,ngq^\beta_{\lambda j}h,\mu)
=x\nu_\beta(n)(y)
=\nu_\beta(n)(xy),
\]
which proves \textup{(CT)}.  If \(\nu_\beta(n)(i,g,\lambda)=(i,g,\lambda)\),
then \(ng=g\), hence \(n=1\); freeness and the final uniqueness statement
follow.

For \textup{(CE)}, write \(x=x_1x_2\) with \(x_1,x_2\in S_\beta\), possible
because the completely simple semigroup \(S_\beta\) satisfies
\(S_\beta=S_\beta^2\).  Since \(ax_1\in S_\beta\), \textup{(CT)} yields
\[
a\nu_\beta(n)(x)
=a\bigl(x_1\nu_\beta(n)(x_2)\bigr)
=(ax_1)\nu_\beta(n)(x_2)
=\nu_\beta(n)(ax).
\]
The right-hand identity is dual.
\end{proof}

\begin{corollary}[Local central-translation form]
\label{cr:cor:local-translation}
Assume \(\tau\) is \(2\)-safe.  If \(x\tau_\beta y\), then there is a unique
\(n\in N_\beta\) such that \(y=\nu_\beta(n)(x)\).
\end{corollary}

\begin{proof}
Write \(x=(i,g,\lambda)\) and \(y=(j,h,\mu)\).  By
Lemma~\ref{cr:lem:index-rigidity}, \(i=j\), \(\lambda=\mu\), and
\(gh^{-1}\in N_\beta\).  Hence \(n:=hg^{-1}\in N_\beta\) and
\(y=\nu_\beta(n)(x)\).  Uniqueness follows from
Lemma~\ref{cr:lem:central-translations}.
\end{proof}

\begin{proposition}[Central-defect homomorphism]
\label{cr:lem:uniform-defect}
Assume \(\tau\) is \(2\)-safe and put
\(\Gamma_\beta:=\{(a,b)\in\Omega_\beta\times\Omega_\beta:a\tau b\}.\)
Then \(\Gamma_\beta\) is a submonoid of \(S\times S\), and there is a unique
map \(z_\beta:\Gamma_\beta\to N_\beta\) such that, for every
\((a,b)\in\Gamma_\beta\) and \(x\in S_\beta\),
\[
bx=\nu_\beta\bigl(z_\beta(a,b)\bigr)(ax),
\qquad
xb=\nu_\beta\bigl(z_\beta(a,b)\bigr)(xa).
\tag{UD}
\]
Moreover \(z_\beta\) is a monoid homomorphism; thus, for all
$\,(a,b),(a',b')\in\Gamma_\beta$,
\begin{equation}
z_\beta(aa',bb')=z_\beta(a,b)z_\beta(a',b').
\label{eq:cr-defect-homomorphism}
\end{equation}
\end{proposition}

\begin{proof}
The set \(\Omega_\beta\) is a submonoid, while compatibility and reflexivity
of \(\tau\) make \(\Gamma_\beta\) a submonoid of \(S\times S\).
Fix \((a,b)\in\Gamma_\beta\).  For \(x\in S_\beta\),
Lemma~\ref{cr:lem:sandwich-ideal} puts \(ax,bx,xa,xb\) in \(S_\beta\), and
compatibility gives \(ax\tau bx\) and \(xa\tau xb\).  Thus
Corollary~\ref{cr:cor:local-translation} gives unique \(n_x,m_x\in N_\beta\)
with
\(bx=\nu_\beta(n_x)(ax)\) and \(xb=\nu_\beta(m_x)(xa)\).

For \(x,y\in S_\beta\), associativity gives \((xb)y=x(by)\).  By
\textup{(CT)}, its two sides are
\(\nu_\beta(m_x)(xay)\) and \(\nu_\beta(n_y)(xay)\), respectively.
Freeness yields \(m_x=n_y\).  Since \(x\) and \(y\) vary independently,
all these elements are one common value, denoted \(z_\beta(a,b)\).  This
proves \textup{(UD)} and uniqueness.

For multiplicativity, put \(z=z_\beta(a,b)\) and
\(z'=z_\beta(a',b')\).  Since \(a'x\in S_\beta\),
\[
\begin{aligned}
(bb')x
&=b\,\nu_\beta(z')(a'x)\\
&=\nu_\beta(z')\bigl(b(a'x)\bigr)\\
&=\nu_\beta(z')\nu_\beta(z)\bigl(a(a'x)\bigr)\\
&=\nu_\beta(zz')\bigl((aa')x\bigr).
\end{aligned}
\]
The second equality is \textup{(CE)}, and the last uses
\(N_\beta\le Z(G_\beta)\).  Uniqueness in \textup{(UD)} gives
\eqref{eq:cr-defect-homomorphism}.  Finally \(z_\beta(1,1)=1\), so \(z_\beta\) is a monoid
homomorphism.
\end{proof}

\begin{corollary}[Trivial defect]
\label{cr:cor:trivial-defect}
Under the hypotheses of Proposition~\ref{cr:lem:uniform-defect},
\(z_\beta(a,b)=1\) if and only if \(ax=bx\) and \(xa=xb\) for every
\(x\in S_\beta\).
\end{corollary}

\begin{proof}
This is immediate from \textup{(UD)} and freeness of the central-translation
action.
\end{proof}

\begin{lemma}[Local faithfulness of the trivial defect]
\label{cr:lem:local-faithfulness}
Assume \(\tau\) is \(2\)-safe, let \(\beta\) be accepting, and let
\(a,b\in S_\beta\) satisfy \(a\tau b\).  If \(z_\beta(a,b)=1\), then
\(a=b\).
\end{lemma}

\begin{proof}
By Lemma~\ref{cr:lem:index-rigidity}, \(a\) and \(b\) lie in the same
\(\mathcal H\)-class.  In a completely simple semigroup this class is a group;
let \(e_H\) be its identity.  Corollary~\ref{cr:cor:trivial-defect} gives
\(be_H=ae_H\), hence \(b=a\).
\end{proof}

\begin{lemma}[Multiplicativity of defects]
\label{cr:lem:defect-product}
Suppose \(a_1\tau b_1,\ldots,a_m\tau b_m\) and \(\beta\) lies below the
support of every displayed element.  Then
\[
z_\beta(a_1\cdots a_m,b_1\cdots b_m)
=\prod_{j=1}^m z_\beta(a_j,b_j).
\]
\end{lemma}

\begin{proof}
Every \((a_j,b_j)\) lies in \(\Gamma_\beta\), so this is the homomorphism
property of \(z_\beta\).
\end{proof}

\begin{corollary}[Contextual defect factorization]
\label{cr:cor:context-defect}
Let \(E=u_0\blank_1u_1\cdots\blank_du_d\) be a tuple context over \(S\),
suppose \(s_i\tau t_i\) for \(1\le i\le d\), and let \(\beta\) lie below
the support of every \(u_j,s_i,t_i\).  Then
\[
z_\beta(E[\mathbf s],E[\mathbf t])
=\prod_{i=1}^d z_\beta(s_i,t_i).
\tag{CD}
\]
\end{corollary}

\begin{proof}
In \(\Gamma_\beta\),
\[
(E[\mathbf s],E[\mathbf t])
=(u_0,u_0)(s_1,t_1)(u_1,u_1)\cdots(s_d,t_d)(u_d,u_d).
\]
Apply the homomorphism \(z_\beta\).  Each fixed pair has defect
\(z_\beta(u_j,u_j)=1\) by Corollary~\ref{cr:cor:trivial-defect}, so only the
coordinate defects remain.
\end{proof}

\subsection{Binary compression of an arbitrary tuple}

\begin{lemma}[Support bookkeeping for tuple contexts]
\label{cr:lem:support-bookkeeping}
Let \(E\) be a tuple context in \(S\).  Write
\(E=u_0\blank_1u_1\cdots\blank_du_d\) and put \(r:=E[s_1,\ldots,s_d]\). Then
\[
\operatorname{supp}(r)
=\operatorname{supp}(u_0)\operatorname{supp}(s_1)\operatorname{supp}(u_1)\cdots
\operatorname{supp}(s_d)\operatorname{supp}(u_d)
\]
in the semilattice \(Y\). In particular,
\(\operatorname{supp}(r)\) lies below the support of every factor occurring
in the filled context. If two tuples have the same product of coordinate
supports, then filling them into the same context produces elements in the
same completely simple component.
\end{lemma}

\begin{proof}
This is repeated application of \textup{(CR1)}: multiplication of elements
from components \(S_\alpha\) and \(S_\delta\) lands in
\(S_{\alpha\delta}\). Since \(Y\) is a semilattice, the product of all
factor supports is below each factor support. The final statement follows
because the fixed-factor contribution is common to both fillings.
\end{proof}

The following lemma uses only the defect calculus already proved, the support
bookkeeping of Lemma~\ref{cr:lem:support-bookkeeping}, and binary separation.
No external completely-regular-semigroup theorem enters from this point on.

\begin{lemma}[Product compression]
\label{cr:lem:product-compression}
Assume $\tau$ is $2$-safe.  Let $d\ge2$ and let
\(\mathbf s=(s_1,\ldots,s_d), \qquad \mathbf t=(t_1,\ldots,t_d)\)
satisfy
\(s_i\tau t_i \qquad(1\le i\le d).\)
If the two tuples share an accepting tuple context, then
\begin{equation}
s_1s_2\cdots s_d
=
t_1t_2\cdots t_d.
\label{eq:cr-product-compression}
\end{equation}
\end{lemma}

\begin{proof}
Let $E$ be a shared accepting tuple context. Compatibility gives \(E[\mathbf s]\tau E[\mathbf t].\) Both values belong to $P$, and $2$-safety includes unary safety.  Applying
Lemma~\ref{cr:lem:syntactic-lifting} at arity one to the related pair
\(E[\mathbf s],E[\mathbf t]\), with the empty unary context, gives
\begin{equation}
E[\mathbf s]
=
E[\mathbf t]
=
p\in P.
\label{eq:cr-shared-accepting-value}
\end{equation}

Put \(\beta:=\operatorname{supp}(p).\) By
Lemma~\ref{cr:lem:support-bookkeeping}, $\beta$ lies below the support of every
tuple component $s_i,t_i$ and every fixed factor occurring in $E$.

Since $p\in P\cap S_\beta$ and $\beta$ lies below the support of every
$s_i,t_i$, Proposition~\ref{cr:lem:uniform-defect} applies to each coordinate pair;
put \(z_i:=z_\beta(s_i,t_i).\)  The two filled products in
\eqref{eq:cr-shared-accepting-value} are the same element $p$.  Because
$\beta=\operatorname{supp}(p)$, the defect $z_\beta(p,p)$ is defined, and
uniqueness in Proposition~\ref{cr:lem:uniform-defect} gives
$z_\beta(p,p)=1$.  Corollary~\ref{cr:cor:context-defect} therefore gives
\begin{equation}
z_1z_2\cdots z_d
=
1_{G_\beta}.
\label{eq:cr-defect-product-one}
\end{equation}
The order here is exactly the fixed Clark--Wurm tuple order; the contextual
factorization corollary has already accounted for all intervening fixed
factors.

Put
\(A:=s_1\cdots s_{d-1}, \qquad B:=s_d, \qquad A':=t_1\cdots t_{d-1}, \qquad B':=t_d.\)
Compatibility gives
\(A\tau A', \qquad B\tau B'.\)
By multiplicativity and \eqref{eq:cr-defect-product-one},
\(z_\beta(AB,A'B')=1.\)
Corollary~\ref{cr:cor:trivial-defect} therefore gives
\(xAB=xA'B'\) for every $x\in S_\beta$.  Hence, for all
$x,y\in S_\beta$,
\begin{equation}
xABy=xA'B'y.
\label{eq:cr-sandwich-equality}
\end{equation}

Moreover, support multiplication is the semilattice product, so
\(\operatorname{supp}(AB)=\prod_{i=1}^d\operatorname{supp}(s_i)\).
The support $\beta=\operatorname{supp}(p)$ is the product of this element
with the supports of the fixed factors occurring in the shared accepting context,
so \(\beta\le\operatorname{supp}(AB).\)
Applying Lemma~\ref{cr:lem:sandwich-ideal} to \(c=AB\) gives
\(S_\beta(AB)S_\beta=S_\beta.\)
Since \(p\in S_\beta\), choose \(x,y\in S_\beta\) such that
\(xABy=p.\)
Then \eqref{eq:cr-sandwich-equality} gives \(xA'B'y=p\) as well.  Thus the coordinatewise
$\tau$-related binary tuples $(A,B)$ and $(A',B')$ are equalized by the
context $E_0=x\blank_1\blank_2y$ at the accepting value $p$.
Lemma~\ref{cr:lem:equalizing-context}, evaluated at the empty-fixed context
$\blank_1\blank_2$, gives \(AB=A'B'.\)  This is exactly \eqref{eq:cr-product-compression}.
\end{proof}

\begin{theorem}[Compatible-tolerance arity-two saturation]
\label{cr:thm:tolerance-saturation}
Let $(S,P)$ be a finite syntactic pointed monoid with $S$ completely regular.
Every $2$-safe compatible tolerance on $S$ is safe at every finite tuple
arity.
\end{theorem}

\begin{proof}
Suppose for a contradiction that $\tau$ fails to be safe at some finite tuple arity.  Since $\tau$ is $2$-safe, such a failure must occur at some arity \(d\ge3\).  Hence there are coordinatewise $\tau$-related tuples
\(\mathbf s=(s_1,\ldots,s_d), \qquad \mathbf t=(t_1,\ldots,t_d)\)
that share an accepting tuple context but have distinct tuple distributions.

By the central-defect homomorphism and its multiplicative consequences,
Lemma~\ref{cr:lem:product-compression} applies and gives
\begin{equation}
c
:=
s_1\cdots s_d
=
t_1\cdots t_d.
\label{eq:cr-saturation-product}
\end{equation}
Consequently their support products agree:
\begin{equation}
\alpha
:=
\operatorname{supp}(s_1)\cdots\operatorname{supp}(s_d)
=
\operatorname{supp}(t_1)\cdots\operatorname{supp}(t_d).
\label{eq:cr-saturation-support}
\end{equation}

Choose a distinguishing tuple context $F$, exchanging the two tuples if
necessary, such that
\begin{equation}
r:=F[\mathbf s]\in P,
\qquad
r':=F[\mathbf t]\notin P.
\label{eq:cr-distinguishing-context}
\end{equation}
Let $\kappa$ be the product in the support semilattice $Y$ of the supports of
the fixed factors occurring in $F$; if there are no fixed factors, take
$\kappa=1_Y=\operatorname{supp}(1_S)$, the empty-product value fixed above. By \eqref{eq:cr-saturation-support} and
Lemma~\ref{cr:lem:support-bookkeeping}, both filled products
lie in the same completely simple component
\begin{equation}
S_\gamma,
\qquad
\gamma:=\kappa\alpha.
\label{eq:cr-gamma-component}
\end{equation}
Moreover
\(\gamma \le \operatorname{supp}(s_i), \operatorname{supp}(t_i) \qquad(1\le i\le d).\)
Since \(r\in P\cap S_\gamma,\) the local rigidity lemmas and the central-defect proposition apply at $\gamma$.

Put \(z_i:=z_\gamma(s_i,t_i).\)  By \eqref{eq:cr-saturation-product} the two ordered
coordinate products are the same element $c$.  Moreover
\(\operatorname{supp}(c)=\alpha, \qquad \gamma=\kappa\alpha\le\alpha,\)
so $c$ acts on the lower component $S_\gamma$ and the defect
$z_\gamma(c,c)$ is defined.  Uniqueness in
Proposition~\ref{cr:lem:uniform-defect} gives $z_\gamma(c,c)=1$.
By multiplicativity of defects,
\begin{equation}
z_1z_2\cdots z_d
=
1_{G_\gamma}.
\label{eq:cr-saturation-defect-one}
\end{equation}

Now evaluate the tuple context $F$.  By \eqref{eq:cr-gamma-component}, $\gamma$ lies below
the support of every fixed factor of $F$ and every tuple coordinate, while
$r\in P\cap S_\gamma$.  Hence
Corollary~\ref{cr:cor:context-defect} applies at the component $S_\gamma$.
Common fixed factors contribute defect $1$, and \eqref{eq:cr-saturation-defect-one} gives
\begin{equation}
z_\gamma(r,r')
=
\prod_{i=1}^d z_i
=
1.
\label{eq:cr-context-defect-one}
\end{equation}
Thus the aggregate coordinate discrepancy is trivial not only for the empty
fixed context used in \eqref{eq:cr-saturation-product}, but also after inserting the fixed factors
of the distinguishing context $F$.

Compatibility gives \(r\tau r'.\)  Since both elements lie in the accepting
component $S_\gamma$, equation~\eqref{eq:cr-context-defect-one} and
Lemma~\ref{cr:lem:local-faithfulness} give $r=r'$, contradicting
\eqref{eq:cr-distinguishing-context}.

Thus $\tau$ is safe at every finite tuple arity. \end{proof}

\begin{theorem}[Completely-regular arity-two saturation]
\label{cr:thm:arity-two-saturation}
Let $L$ be regular and suppose that its syntactic monoid $S$ is completely
regular. Then
\(\cmp_2(L)=\cmp_3(L)=\cmp_4(L)=\cdots.\)
More strongly, every binary-separating relational morphism
$\rho:S\relto M$ is separating at every finite tuple arity.
\end{theorem}

\begin{proof}
Let \(\rho:S\relto M\) be binary-separating.  By
Lemma~\ref{alg:lem:relational-overlap}, its overlap relation \(\tau_\rho\) is a
\(2\)-safe compatible tolerance.  Theorem~\ref{cr:thm:tolerance-saturation}
makes \(\tau_\rho\) safe at every finite arity, and the same lemma converts
this back to all-arity separation of \(\rho\).  The equality
\(\cmp_f(L)=\cmp_2(L)\) for \(f\ge2\) now follows from
Theorem~\ref{alg:thm:relational-characterization}.
\end{proof}

\begin{corollary}[Uniform height bound for completely regular monoids]
\label{cr:cor:height-upper}
For the pseudovariety \(\mathbf{CR}\) of all finite completely regular monoids,
\(\operatorname{ch}(\mathbf{CR})\le 2.\)
\end{corollary}

\begin{proof}
This is the uniform form of
Theorem~\ref{cr:thm:arity-two-saturation}.
\end{proof}

\begin{remark}[Why the proof works beyond orthogroups]
No closure assumption on the idempotents is used. The orthogroup argument
can localize tuple components through products of idempotent supports. In a
general completely regular monoid that shortcut is unavailable. Here it is
replaced by the uniform central translation defect on a lower accepting
completely simple component. Product compression converts one shared
higher-arity acceptance into the exact identity
\(s_1\cdots s_d=t_1\cdots t_d,\)
and the central defect calculus then synchronizes every later tuple context.
This removes the obstruction created by nonorthodox interaction among the
completely simple components.
\end{remark}

\section{Global Compression-Height Trichotomy and Sharp Regimes}
\label{sec:global-arity-trichotomy}

Retain the notation
\(
 M_{ab}=\operatorname{Synt}(\{ab\})\cong\Obj_1(\{ab\}),
 \qquad \mathbf V_{ab}=\langle M_{ab}\rangle,
\)
from Section~\ref{sec:vab-arity}. Let \(U=\{1,u,0\},\qquad u^2=0,\) with \(0\) absorbing.

\begin{lemma}[Every finite non-completely-regular monoid has \(U\) as a divisor]
\label{global:lem:U-divisor}
If a finite monoid \(M\) is not completely regular, then \(U\prec M.\) \end{lemma}

\begin{proof}
Choose \(x\in M\) which belongs to no subgroup and put \(C=\langle x\rangle^1,\qquad I=\{x^n:n\ge2\}.\) Then \(I\) is an ideal of \(C\). Neither \(1\) nor \(x\) lies in \(I\).
Indeed, \(x^n=1\) for some \(n\ge2\) would make \(x\) a unit. If
\(x^n=x\) for some \(n\ge2\), then for \(n=2\) the element \(x\) is
idempotent and hence lies in the trivial subgroup \(\{x\}\).  For
\(n\ge3\), put \(e=x^{n-1}\).  The relation \(x^n=x\) implies
\(e^2=e\), \(ex=xe=x\), and
\(e x^{n-2}=x^{n-2}e=x^{n-2}\), while
\(x x^{n-2}=x^{n-2}x=e\).  Thus \(x^{n-2}\) is an inverse of \(x\)
relative to \(e\), so \(x\) lies in the subgroup of \(C\) with
identity \(e\). Either case contradicts the choice of \(x\). Therefore the Rees quotient \(C/I\) has the three elements \(1,\quad\bar x,\quad0\) with \(\bar x^2=0\), and hence is isomorphic to \(U\).
\end{proof}

\begin{lemma}[Nilpotence plus noncommutativity produces \(M_{ab}\)]
\label{global:lem:mixing}
If \(M\) is a finite noncommutative monoid, then \(M_{ab}\prec U^2\times M.\) \end{lemma}

\begin{proof}
Choose \(g,h\in M\) with \(gh\ne hg\).  In \(U^2\times M\), put
\(\alpha:=((u,1),g)\) and \(\beta:=((1,u),h)\), and let
\(S:=\langle\alpha,\beta\rangle^1\). The projection of a word in
\(\alpha,\beta\) to the first copy of \(U\) is \(1,u,0\) according as the
number of \(\alpha\)'s is \(0,1,\) or at least \(2\); the second copy records
the same trichotomy for the number of \(\beta\)'s. Hence the only elements
represented by words with no zero marker coordinate have occurrence counts
\((0,0),(1,0),(0,1),(1,1)\). They are represented respectively by \(1\), \(\alpha\), \(\beta\), and
\(\alpha\beta\) or \(\beta\alpha\).
The first three are separated by their marker coordinates. The two words of
count \((1,1)\) have the same marker coordinates but distinct third
coordinates \(gh\) and \(hg\), so \(\alpha\beta\ne\beta\alpha\).

Set
\(
I=S\setminus\{1,\alpha,\beta,\alpha\beta\}.
\)
We verify explicitly that \(I\) is a two-sided ideal. Every element of
\(I\setminus\{\beta\alpha\}\) has at least one zero marker coordinate, and
once a coordinate in \(U^2\) is zero it remains zero after multiplication on
either side. The exceptional element \(\beta\alpha\) has one occurrence of
each generator. Multiplying it on either side by any nonidentity element of
\(S\) adds at least one further occurrence of \(\alpha\) or \(\beta\), and
therefore creates a zero in the corresponding marker coordinate.
Multiplication by the identity leaves \(\beta\alpha\) itself in \(I\). Thus
\(SI\cup IS\subseteq I\).

The Rees quotient \(S/I\) therefore has exactly the five elements
\(1\), \(\bar\alpha\), \(\bar\beta\), \(\bar\alpha\bar\beta\), and \(0\).
The marker count shows
\(\bar\alpha^2=\bar\beta^2=0\), while
\(\bar\beta\bar\alpha=0\) because \(\beta\alpha\in I\). Every product of
length at least three repeats at least one generator and hence is \(0\). The
remaining nonzero product is \(\bar\alpha\bar\beta\). This is precisely the
five-element syntactic monoid
\(M_{ab}=\operatorname{Synt}(\{ab\})\).
Thus \(M_{ab}\) is a quotient of the submonoid
\(S\le U^2\times M\), and therefore \(M_{ab}\prec U^2\times M\).
\end{proof}

\begin{theorem}[The \(\mathbf V_{ab}\) structural dichotomy]
\label{global:thm:structural}
Let \(\mathbf V\) be a pseudovariety of finite monoids. If
\(\mathbf V_{ab}\not\subseteq\mathbf V\), then either
\(\mathbf V\subseteq\mathbf{Com}\) or \(\mathbf V\subseteq\mathbf{CR}\). Equivalently, every pseudovariety containing both a noncommutative monoid and
a non-completely-regular monoid contains \(\mathbf V_{ab}\).
\end{theorem}

\begin{proof}
If \(\mathbf V\) is contained in neither class, choose a noncommutative
\(M\in\mathbf V\) and a non-completely-regular \(R\in\mathbf V\).
Lemma~\ref{global:lem:U-divisor} and divisor closure give \(U\in\mathbf V\).
Closure under finite products, submonoids and quotients, together with
Lemma~\ref{global:lem:mixing}, gives \(M_{ab}\in\mathbf V\), and hence
\(\mathbf V_{ab}\subseteq\mathbf V\).
\end{proof}

\begin{remark}[Relation to pseudovariety classifications]
The preceding divisor argument isolates the obstruction needed for SCL
compression height and sits alongside classical classifications of small
noncommutative pseudovarieties. Margolis and Pin classify the
minimal noncommutative varieties, among which the variety generated by
\(\operatorname{Synt}(\{ab\})\) is one canonical aperiodic case
\cite{MargolisPin1984}. A 2025 preprint by Thumm classifies pseudovarieties
satisfying product identities of the form
\(x_1\cdots x_n\approx\rho(x_1,\ldots,x_n)\)~\cite{Thumm2025}; the
published STACS 2026 work of Thumm--Wei\ss concerns straight-line-program
compression~\cite{ThummWeiss2026}.  The classification used here concerns instead the stabilization spectrum of
the Clark--Wurm SCL compression hierarchy.
\end{remark}

\begin{theorem}[Global SCL compression-height trichotomy]
\label{global:thm:trichotomy}
For every pseudovariety \(\mathbf V\) of finite monoids,
\(
\operatorname{ch}(\mathbf V)\in\{1,2,\infty\}.
\)
More precisely,
\[
\operatorname{ch}(\mathbf V)=\infty
\quad\Longleftrightarrow\quad
M_{ab}\in\mathbf V
\quad\Longleftrightarrow\quad
\mathbf V_{ab}\subseteq\mathbf V.
\tag{G}
\]
Consequently no pseudovariety has finite compression height \(3,4,5,\ldots\).
\end{theorem}

\begin{proof}
Suppose first that \(M_{ab}\in\mathbf V\).  Since \(\mathbf V\) is a
pseudovariety, this is equivalent to
\(\mathbf V_{ab}\subseteq\mathbf V\).  By
Corollary~\ref{height:cor:core-bounds},
\(\operatorname{ch}(\mathbf V_{ab})=\infty\), and monotonicity therefore
forces
\(\operatorname{ch}(\mathbf V)=\infty.\)

Conversely, suppose \(M_{ab}\notin\mathbf V\), equivalently
\(\mathbf V_{ab}\not\subseteq\mathbf V\).
Theorem~\ref{global:thm:structural} gives
\(\mathbf V\subseteq\mathbf{Com} \qquad\text{or}\qquad \mathbf V\subseteq\mathbf{CR}.\)
In the first case
Theorem~\ref{group:thm:commutative-collapse} gives
\(\operatorname{ch}(\mathbf V)=1\).
In the second,
Corollary~\ref{cr:cor:height-upper} gives
\(\operatorname{ch}(\mathbf V)\le2\).
Hence every finite value is \(1\) or \(2\), and the infinite value occurs
exactly when \(M_{ab}\in\mathbf V\).
\end{proof}

\begin{remark}[Scope of the trichotomy]
The theorem is a trichotomy of the \emph{possible values} of
\(\operatorname{ch}\), together with a complete characterization of the
infinite value.  It does not assert a complete structural classification of
the two finite cases.  We prove
\(\mathbf V\subseteq\mathbf{Com} \quad\Longrightarrow\quad \operatorname{ch}(\mathbf V)=1\)
and
\(\mathbf V\subseteq\mathbf{CR} \quad\Longrightarrow\quad \operatorname{ch}(\mathbf V)\le2,\)
but the converse characterization of height one, and hence the general
boundary between heights one and two, remains open; see
Section~\ref{sec:discussion}.
\end{remark}

\subsection{Group syntactic monoids: relational compressions collapse to quotients}

\begin{theorem}[Group quotient-collapse]
\label{group:thm:quotient-collapse}
Let \(L\) be regular and suppose that its syntactic monoid is a finite group
\(G\). Then, for every \(f\ge1\), \(\cmp_f(L)=\qcmp_f(L).\) More precisely, from any finite witnessing compression map with \(k\) values one
can extract a witnessing quotient group of \(G\) with at most \(k\) elements.
\end{theorem}

\begin{proof}
Let \(\eta:\Sigma^*\twoheadrightarrow G\) be the syntactic morphism, and let
\(h:\Sigma^*\to M\) be any finite witnessing compression map, with
\(M=\operatorname{im}(h)\). Form the combined image
\(R:=\operatorname{im}(\eta,h)\le G\times M.\)
The first projection \(\pi_1:R\twoheadrightarrow G\) is surjective.

Let \(I\) be the minimal ideal of the finite monoid \(R\). The image
\(\pi_1(I)\) is a nonempty ideal of the group \(G\), hence
\(\pi_1(I)=G\). Choose an idempotent \(e\in I\). Since a group has only
one idempotent, \(\pi_1(e)=1_G\). By the standard structure of minimal
ideals of finite semigroups, \(H:=eIe\) is a group with identity
\(e\)~\cite{Howie1995}. Moreover \(\pi_1(H)=G\): if \(g\in G\), choose
\(x\in I\) with \(\pi_1(x)=g\); then \(exe\in H\) and
\(\pi_1(exe)=g\).

Put
\(\psi:=\pi_2|_H, \qquad H_M:=\psi(H), \qquad K:=\ker\psi.\)
Then \(H_M\) is a finite group contained in \(M\). Since
\(\pi_1|_H:H\twoheadrightarrow G\) is surjective, \(N:=\pi_1(K)\trianglelefteq G.\) There is a well-defined surjective homomorphism
\(H_M\twoheadrightarrow G/N, \qquad \psi(x)\longmapsto \pi_1(x)N.\)
Indeed, if \(\psi(x)=\psi(y)\), then \(x^{-1}y\in K\), so
\(\pi_1(x)^{-1}\pi_1(y)\in N\). Consequently \(|G/N|\le |H_M|\le |M|.\) It remains to show that \(G\to G/N\) separates every unsafe tuple. Suppose
instead that \((\mathbf s,\mathbf t)\in\mathcal U_d(G,P)\), \(d\le f\), is
collapsed in every coordinate, so \(s_i^{-1}t_i\in N\) for all \(i\).
Choose \(k_i\in K\) with \(\pi_1(k_i)=s_i^{-1}t_i\), and choose
\(a_i\in H\) with \(\pi_1(a_i)=s_i\). Then \(\pi_1(a_ik_i)=t_i, \qquad \psi(a_ik_i)=\psi(a_i).\) Because \(H\subseteq R=\operatorname{im}(\eta,h)\), the common second
coordinate \(\psi(a_i)\) belongs to
\(\rho_h(s_i)\cap\rho_h(t_i)\). This happens in every coordinate,
contradicting the fact that \(\rho_h\) is \(f\)-separating by
Theorem~\ref{alg:thm:safety-interface}.

Thus \(G/N\) is a witnessing quotient of size at most \(|M|\). Hence
\(\qcmp_f(L)\le\cmp_f(L)\), and the reverse inequality holds by definition.
\end{proof}

\subsection{A second height-one mechanism: bands}

Commutativity is not necessary for unary collapse.  The entire pseudovariety
of finite bands---idempotent monoids, including highly noncommutative
ones---already saturates at arity one.  We use the standard decomposition of
a band into a semilattice of rectangular \(\mathcal D\)-classes.

\begin{theorem}[Band unary saturation]
\label{band:thm:unary-saturation}
Let \(L\) be regular and suppose that its syntactic monoid \(S\) is a finite
band.  Then every unary-separating relational morphism from \(S\) is
separating at every finite arity.  Consequently
\(\cmp_1(L)=\cmp_2(L)=\cmp_3(L)=\cdots.\)
\end{theorem}

\begin{proof}
Write
\(S=\bigsqcup_{\alpha\in Y}S_\alpha\)
for the standard semilattice decomposition of the band into rectangular
bands, so
\(S_\alpha S_\beta\subseteq S_{\alpha\beta}\).
We first record the elementary absorption that drives the proof.  If
\(a\in S_\alpha\), \(p\in S_\beta\), and \(\beta\le\alpha\), then
\(pap\in S_\beta\).  Since \(S_\beta\) is a rectangular band,
\(p(pap)p=p.\)
But associativity and idempotence of \(p\) make the left side equal to
\(pap\).  Hence
\[
 pap=p.
\tag{B}
\]

Let \(\rho:S\relto M\) be unary-separating and write
\(x\tau y\) when \(\rho(x)\cap\rho(y)\ne\varnothing\).  Suppose two
coordinatewise \(\tau\)-related \(d\)-tuples share an accepting tuple context.
Compatibility gives two \(\tau\)-related accepted filled values.  Unary
separation and syntactic distinguishability force those two values to be one
and the same element \(p\in P\).  Put
\(\beta=\operatorname{supp}(p)\).  Support multiplication in the semilattice
shows that \(\beta\) lies below the support of every tuple coordinate.  Thus
\textup{(B)} gives, for every coordinate pair \(s_i\tau t_i\),
\(p s_i p=p=p t_i p.\)
So \(s_i\) and \(t_i\) share the accepting unary context
\(p\blank_1p\).  If they were distinct, syntacticity would supply another
unary context distinguishing them; hence they would form an unsafe unary pair
collapsed by \(\rho\), contrary to unary separation.  Therefore
\(s_i=t_i\) for every \(i\).  The two tuples are identical and cannot be
unsafe at any arity.  Thus \(\rho\) is separating at every finite arity.
\end{proof}

\subsection{Kernel group languages: arity two detects inner conjugation}

Let \(\eta:\Sigma^*\twoheadrightarrow G\) be surjective and put \(K_\eta:=\eta^{-1}(\{1_G\}).\) \begin{theorem}[Center-index theorem]
\label{group:thm:center-index}
The syntactic monoid of \(K_\eta\) is \(G\), and
\(
\cmp_1(K_\eta)=1,
\qquad
\cmp_f(K_\eta)=[G:Z(G)]=|\operatorname{Inn}(G)|
\quad(f\ge2).
\)
\end{theorem}

\begin{proof}
Distinct group elements are syntactically distinguishable. If \(g\neq h\),
choose a word \(w_g\in\Sigma^*\) with \(\eta(w_g)=g^{-1}\). Appending
\(w_g\) accepts any word of syntactic value \(g\) and rejects any word of
syntactic value \(h\). Hence the syntactic monoid is \(G\).

At arity one, if \(s,t\in G\) share an accepting context, then for some
\(p,q\in G\), \(psq=ptq=1_G\). Cancellation gives \(s=t\). Thus there
are no distinct unsafe unary pairs and \(\cmp_1(K_\eta)=1\).

For the upper bound at every \(f\ge2\), use \(G\twoheadrightarrow G/Z(G)\).
Suppose \(t_i=s_i z_i\) with \(z_i\in Z(G)\). In any tuple context all
\(z_i\) may be moved past the fixed factors and tuple components, so the group
value of the second filling is the value of the first multiplied by
\(z_1\cdots z_d\). If one context accepts both tuples, then
\(z_1\cdots z_d=1_G\). The two fillings therefore have the same group value
in every tuple context. Hence \(\cmp_f(K_\eta)\le [G:Z(G)].\) For the lower bound, by Theorem~\ref{group:thm:quotient-collapse} it suffices
to consider a quotient \(G\to G/N\) witnessing arity two. If
\(g\in N\setminus Z(G)\), choose \(k\in G\) with \(gk\neq kg\), and choose
words \(u,v\in\Sigma^*\) satisfying
\(\eta(u)=k, \qquad \eta(v)=k^{-1}.\)
The tuples
\((1_G,1_G), \qquad (g,g^{-1})\)
are coordinatewise equal modulo \(N\) and share the accepting adjacent-hole
context. But the tuple context
\(F=\blank_1\,u\,\blank_2\,v\)
has syntactic value \(1_G\) on the first tuple and \(gkg^{-1}k^{-1}\neq1_G\) on the second. Thus the tuples are unsafe, a contradiction. Hence
\(N\le Z(G)\), so every witnessing quotient has size at least
\([G:Z(G)]\). Equality follows.
\end{proof}

\begin{corollary}[Reduction of the height-one boundary]
\label{height:cor:height-one-reduction}
Let \(\mathbf V\) be a pseudovariety of finite monoids.
\begin{enumerate}[label=\textup{(\roman*)},leftmargin=*]
\item If \(\operatorname{ch}(\mathbf V)<\infty\), then
\(
 \mathbf V\subseteq\mathbf{Com}
 \text{ or }
 \mathbf V\subseteq\mathbf{CR}.
\)
\item If \(\mathbf V\subseteq\mathbf{CR}\) and \(\mathbf V\) contains a
nonabelian finite group, then \(\operatorname{ch}(\mathbf V)=2\).
\end{enumerate}
Consequently, apart from the already settled commutative case, the unresolved
height-one classification is confined to pseudovarieties
\(\mathbf V\subseteq\mathbf{CR}\) whose group members are all abelian.
\end{corollary}

\begin{proof}
For \textup{(i)}, finite height and
Theorem~\ref{global:thm:trichotomy} imply \(M_{ab}\notin\mathbf V\), so
Theorem~\ref{global:thm:structural} gives the stated dichotomy.
For \textup{(ii)}, Corollary~\ref{cr:cor:height-upper} gives the upper bound
\(\operatorname{ch}(\mathbf V)\le2\).  Let \(G\in\mathbf V\) be a
nonabelian finite group and choose a finite alphabet with a surjection
\(\eta:\Sigma^*\twoheadrightarrow G\).  The kernel language \(K_\eta\)
has syntactic monoid \(G\), while
Theorem~\ref{group:thm:center-index} gives
\(
 \cmp_1(K_\eta)=1< [G:Z(G)]=\cmp_2(K_\eta),
\)
because a nonabelian group has proper center.  Thus height one is impossible,
and the upper bound is sharp.
\end{proof}

\subsection{Height sharpness and classical benchmarks}

\begin{corollary}[The pseudovariety of bands has height one]
\label{band:cor:height-one}
Let \(\mathbf B\) be the pseudovariety of all finite bands.  Then
\(\operatorname{ch}(\mathbf B)=1.\)
In particular, height one does not imply commutativity.  The canonical
minimal noncommutative one-sided band pseudovarieties
\(\langle U_2\rangle\) and \(\langle U_2^{\mathrm{op}}\rangle\) are
contained in \(\mathbf B\) and therefore also have height one.
\end{corollary}

\begin{proof}
The height statement is Theorem~\ref{band:thm:unary-saturation}.  The
three-element monoid \(U_2=\operatorname{Synt}(A^*a)\) is a noncommutative
band, and its opposite is the syntactic monoid of \(aA^*\)
\cite{MargolisPin1984}.
\end{proof}

\begin{corollary}[Sharpness of the completely-regular bound]
\label{cr:cor:height}
For the pseudovariety $\mathbf{CR}$ of all finite completely regular monoids,
\(
\operatorname{ch}(\mathbf{CR})=2.
\)
\end{corollary}

\begin{proof}
Theorem~\ref{cr:thm:arity-two-saturation} gives the upper bound. The
pseudovariety $\mathbf{CR}$ contains nonabelian finite groups.  For example,
$Z(S_3)=1$, so the kernel language for the syntactic group $S_3$ has
\(\cmp_1=1, \qquad \cmp_f=6\quad(f\ge2)\)
by the center-index theorem.  Thus a strict arity-$1$ to arity-$2$ jump occurs,
and the height is exactly two.
\end{proof}

\subsection{Benchmark pseudovarieties and minimal obstructions}

\begin{corollary}[Benchmark pseudovarieties]
\label{global:cor:benchmarks}
With the standard notation \(\mathbf G\) for all finite groups,
\(\mathbf A\) for all finite aperiodic monoids, and
\(\mathbf J,\mathbf R,\mathbf L\) for the \(\mathcal J\)-,
\(\mathcal R\)-, and \(\mathcal L\)-trivial pseudovarieties, the following
values hold:
\[
\begin{array}{c|cccccccc}
\mathbf V
&\mathbf{Com}&\mathbf B&\mathbf G&\mathbf{CR}
&\mathbf A&\mathbf J&\mathbf R&\mathbf L\\
\hline
\operatorname{ch}(\mathbf V)
&1&1&2&2&\infty&\infty&\infty&\infty
\end{array}
\]
\end{corollary}

\begin{proof}
The first value is the commutative collapse and the second is
Corollary~\ref{band:cor:height-one}.  Since groups and completely regular
monoids have height at most two, while any nonabelian finite group kernel has
\(\cmp_1=1<\cmp_2\) by Theorem~\ref{group:thm:center-index}, both
\(\mathbf G\) and \(\mathbf{CR}\) have height exactly two.

The monoid \(M_{ab}\) is aperiodic and \(\mathcal J\)-trivial (its principal
two-sided ideals are all distinct).  Hence
\(M_{ab}\in\mathbf A\cap\mathbf J\), and
\(\mathbf J\subseteq\mathbf R\cap\mathbf L\).  The infinite values therefore
follow from Theorem~\ref{global:thm:trichotomy}.
\end{proof}

The benchmark table also locates the classical Margolis--Pin minimal
noncommutative frontier~\cite{MargolisPin1984}: the two one-sided band types
have height one, the noncommutative group type has height two, and
\(\langle M_{ab}\rangle\) has infinite height.  Thus all three values already
occur on that minimal noncommutative frontier.

\section{Graph Realization and Computational Complexity}
\label{sec:graph-universality}

The arity-one framework began with the contextual-conflict graph \(\Gamma_L\).
We now show that this graph geometry is universal: every nonempty finite simple graph is
realized as an induced subgraph on designated vertex symbols of a finite
length-three language, while every other conflict-graph vertex is isolated.
The resulting equality with chromatic number is therefore a representation
theorem first and an NP-hardness reduction second.

\subsection{Graph-language construction}

Let \(G=(V,E)\) be a finite simple graph.

For every vertex \(v\in V\), introduce a terminal \(x_v.\) For every edge \(e\in E\), introduce two terminals \(\ell_e,\qquad r_e.\) For every vertex \(v\in V\), introduce two additional private terminals \(p_v,\qquad q_v.\) All these terminals are pairwise distinct.

For an edge \(e=\{u,v\},\) put the two edge words
\(\ell_e x_u r_e, \qquad \ell_e x_v r_e\)
into the language.
For every vertex \(v\), also put the private word \(p_vx_vq_v\) into the language.

Thus
\(
L_G
:=
\{
\ell_ex_ur_e,\ell_ex_vr_e:
e=\{u,v\}\in E
\}
\cup
\{
p_vx_vq_v:v\in V
\}.
\)
Every word of \(L_G\) has length exactly three, and \(|L_G|=2|E|+|V|.\) 

\subsection{Classification of unsafe factors}

\begin{lemma}[Classification of shared unary contexts]
\label{graph:lem:factor-classification}
If two factors $s,t$ share an accepting unary context for $L_G$, then
$|s|=|t|\le3$.  Moreover:
\begin{enumerate}[label=\textup{(\roman*)},leftmargin=*]
\item lengths $0$ and $3$ create no distinct unsafe pair;
\item distinct length-two factors that share an accepting context have equal
unary distributions;
\item the length-one distributions are
\[
\D_{L_G}(x_v)=\{(\ell_e,r_e):e\ni v\}\cup\{(p_v,q_v)\},
\]
\[
\D_{L_G}(\ell_e)=\{(\varepsilon,x_ur_e),(\varepsilon,x_vr_e)\},
\qquad
\D_{L_G}(r_e)=\{(\ell_ex_u,\varepsilon),(\ell_ex_v,\varepsilon)\}
\]
for $e=\{u,v\}$, while
\(
\D_{L_G}(p_v)=\{(\varepsilon,x_vq_v)\}
\)
and
\(
\D_{L_G}(q_v)=\{(p_vx_v,\varepsilon)\}.
\)
\end{enumerate}
\end{lemma}

\begin{proof}
If $psq,ptq\in L_G$, then both completed words have length three, so
$|s|=|t|\le3$ and both $s,t$ are factors of accepted words.
For length $0$ there is only $\varepsilon$, while every accepted length-three
word has distribution $\{(\varepsilon,\varepsilon)\}$, proving \textup{(i)}.

For \textup{(ii)}, the only nonprivate shared length-two pairs are, for
$e=\{u,v\}$,
\[
\D_{L_G}(\ell_ex_u)=\D_{L_G}(\ell_ex_v)=\{(\varepsilon,r_e)\},
\qquad
\D_{L_G}(x_ur_e)=\D_{L_G}(x_vr_e)=\{(\ell_e,\varepsilon)\}.
\]
The private factors $p_vx_v$ and $x_vq_v$ have the singleton distributions
$\{(\varepsilon,q_v)\}$ and $\{(p_v,\varepsilon)\}$, respectively.
Distinct gadgets cannot share their unique completing marker, and a prefix
factor cannot share a context with a suffix factor because the hole position
is fixed.  This proves \textup{(ii)}.  Finally, \textup{(iii)} follows by
listing the occurrences of each terminal in the defining words of $L_G$.
\end{proof}

\begin{lemma}[Unsafe-pair graph is exactly \(G\)]
\label{graph:lem:unsafe-exact}
The distinct unsafe pairs of \(L_G\) are exactly
\(\{x_u,x_v\} \quad\text{with}\quad \{u,v\}\in E.\)
More precisely, for distinct vertices \(u,v\), \(x_u,x_v\text{ are unsafe} \iff \{u,v\}\in E,\) and no pair involving a non-vertex terminal is unsafe.
\end{lemma}

\begin{proof}
Suppose first that \(e=\{u,v\}\in E.\) Then the context \(\ell_e\blank_1r_e\) accepts both \(x_u\) and \(x_v\), so their distributions intersect.
However, \((p_u,q_u)\in\D_{L_G}(x_u)\) and, by privacy of the markers, \((p_u,q_u)\notin\D_{L_G}(x_v).\) Thus the distributions are different, and \(x_u,x_v\) are unsafe.

Conversely, if \(u\neq v\) and
\(\D_{L_G}(x_u)\cap\D_{L_G}(x_v)\neq\varnothing,\)
Lemma~\ref{graph:lem:factor-classification}\textup{(iii)} shows that a common context must be \((\ell_e,r_e)\) for an edge \(e\) incident with both \(u\) and \(v\). Since the graph is
simple, \(e=\{u,v\}.\) It remains to exclude pairs involving marker terminals.
Every \(x_v\) occurs only in the middle position of accepted words.
The terminals \(\ell_e,p_v\) occur only in the first position, while \(r_e,q_v\) occur only in the third position.
A shared unary context fixes the position of the hole, so no \(x_v\) can share
an accepting context with a marker terminal.

Two first-position markers share an accepting context only if the same
length-two suffix follows both. The suffixes following \(\ell_e\) end in the
private edge marker \(r_e\), whereas the suffix following \(p_v\) ends in the
private marker \(q_v\); markers belonging to distinct gadgets are pairwise
different. Hence distinct first-position markers share no accepting context.
The same argument applies to third-position markers by using their length-two
prefixes.

Together with Lemma~\ref{graph:lem:factor-classification}\textup{(i)--(ii)}, this exhausts all possible factor lengths.
\end{proof}

\begin{corollary}[Every finite graph occurs as an induced contextual-conflict geometry]
\label{graph:cor:conflict-realization}
For the language \(L_G\), the subgraph of \(\Gamma_{L_G}\) induced by the
designated vertex terminals \(\{x_v:v\in V(G)\}\) is isomorphic to
\(G\), and every other vertex of \(\Gamma_{L_G}\) is isolated.  Equivalently,
\(\Gamma_{L_G}\) is the disjoint union of this induced copy of \(G\) and an
independent set of isolated vertices.
\end{corollary}

\begin{proof}
Lemma~\ref{graph:lem:unsafe-exact} identifies all unsafe pairs between factors
of accepted words.  Words that are not factors of accepted words have empty
distribution and hence participate in no unsafe pair.  The claim follows from
the definition of \(\Gamma_L\) in Section~\ref{sec:conflict}.

Moreover, the designated words \(x_v\) lie in pairwise distinct syntactic
classes, since the private context \((p_v,q_v)\) belongs to
\(\D_{L_G}(x_v)\) and to no \(\D_{L_G}(x_u)\) with \(u\ne v\).
Hence Corollary~\ref{conf:cor:omega-syntactic} also realizes \(G\) as an
induced subgraph of the unary principal-overlap graph \(\Om_1(L_G)\).
\end{proof}

\subsection{Exact realization of chromatic number}

The significance of the construction is the exact identity with chromatic
number, not merely the resulting hardness statement. In particular, the
compression classes of the distinguished length-one factors are exactly graph
color classes; the remaining factors are arranged so that they create no
additional unsafe identifications. The alphabet grows with the graph. The
corresponding fixed-binary-alphabet problem is left open.

\begin{lemma}[Every compression map induces a proper coloring]
\label{graph:lem:lower-coloring}
If \(h:\Sigma_G^*\to M\) witnesses unary substitutability of \(L_G\), then \(c_h:V\to\operatorname{im}(h), \qquad c_h(v):=h(x_v),\) is a proper coloring of \(G\).
Consequently
\(
\chi(G)\le|\operatorname{im}(h)|.
\)
\end{lemma}

\begin{proof}
For every edge \(\{u,v\}\in E,\) Lemma~\ref{graph:lem:unsafe-exact} says that \(x_u,x_v\) are unsafe.
A witnessing compression map must assign different values to every unsafe pair,
so \(h(x_u)\neq h(x_v).\) Thus \(c_h\) is proper.
\end{proof}

\begin{lemma}[Every proper coloring gives a compression map]
\label{graph:lem:upper-coloring}
If \(G\) has a proper coloring with \(q\ge2\) colors, then there is a
\(q\)-element monoid compression map witnessing unary substitutability of \(L_G\).
If \(G\) is edgeless, the one-element compression is a witness.
\end{lemma}

\begin{proof}
Assume first \(q\ge2\).
Let \(M_q=\{1,0,c_1,\ldots,c_{q-2}\}.\) Define multiplication by declaring \(1\) to be the identity, \(0\) to be
absorbing, and \(mn=0\) whenever \(m,n\neq1\).
This is an associative monoid with exactly \(q\) elements.

Choose a proper coloring \(\gamma:V\to M_q\) using the \(q\) monoid elements as color names, and define the compression map on
generators by \(h(x_v):=\gamma(v),\) while every marker terminal \(\ell_e,r_e,p_v,q_v\) is sent to \(0\).
By freeness of \(\Sigma_G^*\), this extends uniquely to a monoid homomorphism \(h:\Sigma_G^*\to M_q.\) By Lemma~\ref{graph:lem:unsafe-exact}, every unsafe pair has the form \(x_u,x_v\) for an edge \(\{u,v\}\).
Properness of \(\gamma\) gives \(h(x_u)\neq h(x_v).\) Hence no unsafe pair receives the same compression value, so \(h\) witnesses
unary substitutability.

If \(G\) is edgeless, Lemma~\ref{graph:lem:unsafe-exact} says that \(L_G\) has no
unsafe pair at all. Therefore the unique homomorphism to the one-element
monoid is a witness.
\end{proof}

\begin{theorem}[Exact graph realization]
\label{graph:thm:graph-realization}
For every nonempty finite simple graph \(G\), \(\cmp_1(L_G)=\chi(G).\) Moreover, the construction \(G\mapsto L_G\) is polynomial time, all words of
\(L_G\) have length three, and \(|L_G|=2|E|+|V|.\) \end{theorem}

\begin{proof}
Lemma~\ref{graph:lem:lower-coloring} gives \(\chi(G)\le\cmp_1(L_G).\) A minimum proper coloring, together with
Lemma~\ref{graph:lem:upper-coloring}, gives \(\cmp_1(L_G)\le\chi(G).\) The construction size and word-length assertions are immediate from the
definition of \(L_G\).
\end{proof}

\begin{remark}[Sharpness of the geometric lower bound]
Theorem~\ref{graph:thm:graph-realization} shows that the general geometric
lower bound of Proposition~\ref{scl:prop:chromatic-lower} is already sharp
at unary arity.  Indeed, under the canonical unary identification,
\(\chi(\Om_1(L_G))=\chi(G)=\cmp_1(L_G).\)
Thus arbitrary finite conflict geometry can attain the compositional
compression lower bound exactly.
\end{remark}

\begin{remark}
The alphabet \(\Sigma_G\) is part of the explicit finite-language input. The
corresponding variable-threshold problem over a fixed binary alphabet remains
open; the exact realization theorem concerns the natural variable-alphabet
representation of the explicit finite language.
\end{remark}

\subsection{The finite-language decision problem}

\begin{definition}
Let \({\normalfont\textsc{Length-3-Obs}}_1(3)\) be the following decision problem.

\medskip
\noindent
\textbf{Input:} An explicitly listed finite language \(L\subseteq\Sigma^3,\) where the finite alphabet \(\Sigma\) is part of the input.

\noindent
\textbf{Question:} \(\cmp_1(L)\le3?\) \end{definition}

\begin{lemma}[Polynomial verification for fixed threshold]
\label{graph:lem:np}
The problem
\(
{\normalfont\textsc{Length-3-Obs}}_1(3)
\)
belongs to NP.
\end{lemma}

\begin{proof}
A certificate gives a multiplication table for a monoid \(M\) with at most
three elements, its identity, and the value \(h(a)\in M\) of every terminal.
The table is constant size, the monoid laws are checkable directly, and the
generator assignment determines a homomorphism \(h:\Sigma^*\to M\).

Only factors of listed words have nonempty unary distributions. Since all
accepted words have length three, there are polynomially many such factors.
For each factor \(x\), enumerate its complete distribution
\(\D_L(x)=\{(p,q):pxq\in L\}.\)
Then check every factor pair \(x,y\): whenever \(h(x)=h(y)\) and the two
distributions intersect, they must be equal. This is exactly unary
substitutability and is polynomial in the explicit input size.
\end{proof}

\begin{remark}[Why the NP argument is stated at a fixed threshold]
The preceding certificate argument uses that the threshold is fixed: for
\(k=3\), a multiplication table for the witness monoid has constant size.
If \(k\) is itself part of the input, an explicit multiplication table has
\(\Theta(k^2\log k)\) bits.  Without an a priori polynomial bound on the
relevant witness size in terms of the language encoding, the same argument
does not by itself place the variable-threshold problem in NP.  This issue is
included in the regular-input complexity problem below.
\end{remark}

\begin{theorem}[NP-completeness for explicitly listed length-three languages]
\label{graph:thm:npcomplete}
When the finite alphabet is part of the input, the problem
\({\normalfont\textsc{Length-3-Obs}}_1(3)\) is NP-complete.
\end{theorem}

\begin{proof}
Membership in NP is Lemma~\ref{graph:lem:np}.

For NP-hardness, reduce from {\normalfont\textsc{3-Colorability}} \cite{GareyJohnson1979}.
Given a finite simple graph \(G\), if \(V(G)=\varnothing\), first add one isolated
vertex; this preserves 3-colorability.  Thus we may assume that \(G\) is nonempty
and construct \(L_G\) as above.  This construction is polynomial time and yields
\(L_G\subseteq\Sigma_G^3.\) By Theorem~\ref{graph:thm:graph-realization},
\(\cmp_1(L_G)=\chi(G).\) Therefore
\(G\text{ is \(3\)-colorable} \iff \chi(G)\le3 \iff \cmp_1(L_G)\le3.\)
Hence the reduction is correct.
\end{proof}

\subsection{Complexity from an explicit pointed monoid}

The preceding hardness is not an artifact of using a listed finite language as
the input representation.  The finite algebraic optimization itself is already
NP-complete.

\begin{definition}
Let \({\normalfont\textsc{Pointed-Obs}}_1\) be the following decision problem.

\medskip
\noindent
\textbf{Input:} A finite monoid \(T\) by its multiplication table, an accepting
subset \(P\subseteq T\), and a positive integer \(k\).

\noindent
\textbf{Question:} Is there a unary-separating relational morphism
\(\rho:T\relto M\) to some finite monoid \(M\) with \(|M|\le k\)?
\end{definition}

\begin{theorem}[NP-completeness for pointed-monoid input]
\label{graph:thm:pointed-npcomplete}
The problem \({\normalfont\textsc{Pointed-Obs}}_1\) is NP-complete.  NP-hardness already
holds for the fixed threshold \(k=3\), even when the input is promised to be a
syntactic pointed monoid of a finite length-three language.
\end{theorem}

\begin{proof}
For membership in NP, put \(n=|T|\).  If \(k\ge n\), the identity morphism is
a witness, so assume \(k<n\).  A certificate consists of a multiplication
table for a monoid \(M\) of size at most \(k\), its identity, and the relation
\(\rho\subseteq T\times M\).  Its size is polynomial in the input because
\(k<n\).  The monoid laws, nonempty fibers, identity condition, and inclusions
\(\rho(s)\rho(t)\subseteq\rho(st)\) are polynomial-time checkable.
For each \(s\in T\), enumerate
\(\Phi_1(s)=\{(q_0,q_1)\in T^2:q_0sq_1\in P\}.\)
This identifies all unary unsafe pairs in polynomial time, after which the
condition
\(\rho(s)\cap\rho(t)=\varnothing\) for every unsafe pair is also polynomial.
Thus the problem is in NP.

For NP-hardness, start with a finite simple graph \(G\); if \(V(G)=\varnothing\),
first add one isolated vertex, which preserves 3-colorability.  Hence we may assume
that \(G\) is nonempty and construct the length-three language \(L_G\) of
Theorem~\ref{graph:thm:graph-realization}.  Its pointed
syntactic monoid can be constructed in polynomial time.  Indeed, every word
with nonempty unary distribution is a contiguous factor of a listed
length-three word; there are only polynomially many such factors, while all
nonfactors have empty distribution and form one zero class.  Equality of
syntactic classes is equality of the explicitly enumerable unary
distributions, and multiplication of classes is obtained by concatenating
representatives and looking up the resulting factor class (or the zero class).
Hence the full multiplication table and accepting subset are computable in
polynomial time and have polynomial size.

By Theorem~\ref{alg:thm:relational-characterization} and syntactic invariance,
the least codomain size of a unary-separating relational morphism from this
pointed monoid is exactly \(\cmp_1(L_G)=\chi(G)\).  Therefore \(G\) is
3-colorable exactly when
\({\normalfont\textsc{Pointed-Obs}}_1(T_{L_G},P_G,3)\) is a yes-instance.
The reduction is polynomial, proving NP-hardness.
\end{proof}

\section{Discussion and Open Problems}
\label{sec:discussion}

The central phenomenon is the contrast between local richness and global
rigidity.  Inside \(\mathbf V_{ab}\), individual languages realize arbitrary
finite strict arity prefixes and unbounded successive gaps.  After
uniformizing over a pseudovariety, however,
\(\operatorname{ch}(\mathbf V)\in\{1,2,\infty\}, \qquad \operatorname{ch}(\mathbf V)=\infty \Longleftrightarrow M_{ab}\in\mathbf V.\)
The completely regular saturation theorem explains the finite side: pointed
binary separation reduces discrepancies to multiplicative central defects,
and a shared higher-arity acceptance forces their aggregate defect to vanish.

The SCL formulation also separates geometric complexity from the size of the
least reusable compositional resource.  Nonabelian group kernels can have
rapidly growing principal tuple geometry while the optimum stabilizes at
arity two; at unary arity, graph realization gives arbitrary finite conflict
geometry with \(\cmp_1(L_G)=\chi(G)\).  Independently, the strict unary
examples show that forcing the compression to factor through a syntactic
quotient can have unbounded cost.  The trichotomy therefore concerns the
stabilization of finite compositional separation, not stabilization of the
concept lattices or of quotient complexity.

Two problems seem particularly natural.
\begin{enumerate}[leftmargin=*]
\item \textbf{The remaining height-one boundary inside \(\mathbf{CR}\).}
Corollary~\ref{height:cor:height-one-reduction} reduces the unresolved part of
the classification to completely regular pseudovarieties whose group members
are all abelian.  Is it true that
\(\mathbf V\subseteq\mathbf{CR} \quad\text{and every group in \(\mathbf V\) is abelian} \quad\Longrightarrow\quad \operatorname{ch}(\mathbf V)=1\,?\)
Bands give the extreme case in which all maximal subgroups are trivial, while
commutative monoids give another positive regime.  A concrete first family to
test is provided by finite monoids obtained by adjoining an identity to
completely simple Rees matrix semigroups
\(\mathcal M[A;I,\Lambda;Q]\) over finite abelian groups \(A\), and by the
pseudovarieties they generate.

\item \textbf{Regular-input complexity beyond the pointed-monoid model.}
Theorem~\ref{graph:thm:pointed-npcomplete} settles the unary problem for an
explicit pointed monoid.  Determine the complexity from a DFA, the behavior
when the arity bound varies, and whether explicit syntactic-monoid blow-up can
be avoided.

\end{enumerate}

\section*{Statements and Declarations}

\noindent\textbf{Author contribution.}
The author conceived the study, developed the mathematical results, verified
the proofs, and prepared the manuscript.

\noindent\textbf{Funding.}
The author received no funding for this work.

\noindent\textbf{Competing interests.}
The author declares no competing interests.

\noindent\textbf{Data availability.}
No datasets were generated or analyzed for this theoretical study.

\noindent\textbf{Use of generative AI.}
During manuscript development, the author used OpenAI ChatGPT for assistance
with literature searching, proof auditing, organization, and language editing.
The author remains fully responsible for the mathematical statements, proofs,
citations, and final manuscript.

\end{document}